\documentclass[useAMS,referee,usenatbib]{biom}

\usepackage{amsmath}
\usepackage{amssymb}
\usepackage{mathrsfs}
\usepackage{xcolor}
\usepackage{multirow,booktabs}
\usepackage{graphicx}
\usepackage{hyperref}
\usepackage{multibib}

\newcites{supporting}{References}

\def\bSig\mathbf{\Sigma}

\def\bSigtld{\widetilde{\mathbf{\Sigma}}}
\def\bmutld{\widetilde{\boldsymbol\mu}}

\def\trans{^{\sf \tiny T}}

\def\bzero{\boldsymbol 0}
\def\bone{\boldsymbol 1}

\def\Gscr{\mathcal{G}}
\def\Dscr{\mathscr{D}}
\def\Wscrhat{\widehat{\mathcal{W}}}
\def\Wscrtld{\widetilde{\mathcal{W}}}
\def\Vscrhat{\widehat{\mathcal{V}}}
\def\Mscr{\mathcal{M}}
\def\Mscrpi{\mathcal{M}_{\pi}}
\def\Mscrmu{\mathcal{M}_{\mu}}
\def\Mscrnp{\mathcal{M}_{np}}
\def\Mscrpitld{\widetilde{\Mscr}_{\pi}}
\def\Mscrmutld{\widetilde{\Mscr}_{\mu}}
\def\Ascr{\mathcal{A}}

\def\Ascrpi{\mathcal{A}_{\balph}}
\def\Ascrmu{\mathcal{A}_{\bbeta}}
\def\Ascrmuk{\mathcal{A}_{\bbeta_k}}

\def\Ascrpitru{\mathcal{A}_{\pi}}
\def\Ascrmutru{\mathcal{A}_{\mu}}
\def\Xscr{\mathcal{X}}
\def\Escr{\mathcal{E}}
\def\Sscr{\mathcal{S}}

\def\Lscr{\mathcal{L}}

\def\bX{\mathbf{X}}
\def\bx{\mathbf{x}}
\def\bXtld{\widetilde{\bX}}

\def\bZ{\mathbf{Z}}
\def\bXdagT{\bX^{\dagger\sf \tiny T}}
\def\bXddagT{\bX^{\ddagger\sf \tiny T}}
\def\bSbar{\bar{\mathbf{S}}}
\def\bShat{\widehat{\mathbf{S}}}
\def\bV{\mathbf{V}}
\def\bu{\mathbf{u}}
\def\bs{\mathbf{s}}
\def\bv{\mathbf{v}}
\def\lhat{\widehat{l}}
\def\fhat{\widehat{f}}
\def\bm{\boldsymbol m}

\def\ba{\boldsymbol a}

\def\bU{\boldsymbol U}
\def\hhat{\widehat{h}}

\def\bI{\boldsymbol I}

\def\Delthat{\widehat{\Delta}}
\def\Delthatdr{\widehat{\Delta}_{dr}}
\def\Deltbar{\bar{\Delta}}
\def\Deltstr{\Delta^*}
\def\alphtld{\widetilde{\alpha}}
\def\alphhat{\widehat{\alpha}}
\def\alphbar{\bar{\alpha}}
\def\balph{\boldsymbol\alpha}
\def\balphbar{\bar{\boldsymbol\alpha}}
\def\balphhat{\widehat{\balph}}

\def\balphvec{\vec{\balph}}

\def\betatld{\widetilde{\beta}}

\def\bbetavec{\vec{\bbeta}}
\def\betahat{\widehat{\beta}}
\def\betabar{\bar{\beta}}
\def\bbeta{\boldsymbol\beta}
\def\bbetabar{\bar{\boldsymbol\beta}}
\def\bbetahat{\widehat{\boldsymbol\beta}}

\def\Thetpi{\Theta_{\balph}}
\def\Thetmu{\Theta_{\bbeta}}
\def\bthet{\boldsymbol\theta}
\def\bthetbar{\bar{\boldsymbol\theta}}
\def\bthethat{\widehat{\boldsymbol\theta}}
\def\muhat{\widehat{\mu}}

\def\mubar{\bar{\mu}}
\def\mutld{\widetilde{\mu}}
\def\bepsi{\boldsymbol\varepsilon}
\def\bpsi{\boldsymbol\psi}
\def\bdeta{\boldsymbol\eta}
\def\sighat{\widehat{\sigma}}

\def\bPsi{\boldsymbol\Psi}
\def\bUpsi{\boldsymbol\Upsilon}
\def\bzeta{\boldsymbol\zeta}

\def\ddbalph{\frac{\partial}{\partial\balph}}
\def\ddbalphT{\frac{\partial}{\partial\balph\trans}}
\def\ddbbeta{\frac{\partial}{\partial\bbeta}}
\def\ddbbetaT{\frac{\partial}{\partial\bbeta\trans}}
\def\ddbs{\frac{\partial}{\partial\boldsymbol s}}
\def\ddbsT{\frac{\partial}{\partial\boldsymbol s\trans}}

\def\ddbssq{\frac{\partial}{\partial\boldsymbol s^{\otimes 2}}}
\def\ddbsq{\frac{\partial}{\partial\boldsymbol s^{\otimes q}}}
\def\ddpsitwoi{\frac{\partial}{\partial\psi_{2i}}}

\def\pihat{\widehat{\pi}}
\def\ipwhati{\frac{I(T_i=k)}{\pihat_k(\bX_i;\bthethat)}} 
\def\ipwbari{\frac{I(T_i=k)}{\pi_k(\bX_i;\bthetbar)}}

\def\Kdot{\dot{K}}

\def\E{\mathbb E}
\def\P{\mathbb P}

\def\indep{\perp\!\!\!\perp}

\newcommand{\abs}[1]{\left|#1\right|}
\newcommand{\norm}[1]{\left\lVert#1\right\rVert}

\DeclareMathOperator*{\argmax}{arg\,max}

\definecolor{darkred}{RGB}{150,50,50}

\newenvironment{eq} 
{
\begin{equation}
\begin{aligned} 
}
{
\end{aligned}
\end{equation}
}
\newenvironment{eq*} 
{
\begin{equation*}
\begin{aligned} 
}
{
\end{aligned}
\end{equation*}
}

\newenvironment{hproof}{\emph{Proof sketch:}}{}

\title[Estimating ATE with a Double-Index Propensity Score]{Estimating Average Treatment Effects with a Double-Index Propensity Score}

\author{David Cheng$^{1}$,
Abhishek Chakrabortty$^{2}$, Ashwin N. Ananthakrishnan$^{3}$, and Tianxi Cai$^{4,*}$\email{tcai@hsph.harvard.edu} \\
$^{1}$VA Boston Healthcare System, Boston, Massachusetts \\
$^{2}$Department of Statistics, Texas A\&M University, College Station, Texas \\
$^{3}$Division of Gastroenterology, Massachusetts General Hospital, Boston, Massachusetts\\
$^{4}$Department of Biostatistics, Harvard T.H. Chan School of Public Health, Boston, Massachusetts}

\begin{document}

\label{firstpage}

\begin{abstract}
We consider estimating average treatment effects (ATE) of a binary treatment in observational data when data-driven variable selection is needed to select relevant covariates from a moderately large number of available covariates $\bX$.
To leverage covariates among $\bX$ predictive of the outcome for efficiency gain while using regularization to fit a parametric propensity score (PS) model, we consider a dimension reduction of $\bX$ based on fitting both working PS and outcome models using adaptive LASSO.  A novel PS estimator, the Double-index Propensity Score (DiPS), is proposed, in which the treatment status is smoothed over the linear predictors for $\bX$ from both the initial working models.  The ATE is estimated by using the DiPS in a normalized inverse probability weighting (IPW) estimator, which
is found to maintain double-robustness and also local semiparametric efficiency with a fixed number of covariates $p$.
Under misspecification of working models, the smoothing step leads to gains in efficiency and robustness over traditional doubly-robust estimators.
These results are extended to the case where $p$ diverges with sample size and working models are sparse.  Simulations show the benefits of the approach in finite samples.  We illustrate the method by estimating the ATE of statins on colorectal cancer risk in an electronic medical record (EMR) study and the effect of smoking on C-reactive protein (CRP) in the Framingham Offspring Study.
\end{abstract}


\maketitle

\section{Introduction}
There is growing interest in evaluating medical treatments and policies in large-scale observational data such as electronic medical records (EMR).
As with any observational data, in the absence of randomization, adjustment for a sufficient set of pre-treatment covariates $\bX$ that satisfy ``no unmeasured confounding'' 
is needed when estimating average treatment effects (ATE) to avoid confounding bias.  This is routinely done using propensity score (PS), outcome regression, and doubly-robust (DR) methods 
\citep{lunceford2004stratification}. 
These methods were initially developed in settings where $p$, the dimension of $\bX$, was small relative to the sample size $n$. 
But large-scale observational data are increasingly collecting rich measurements in large sets of covariates, and data-driven variable selection approaches are needed due to the lack of sufficient prior knowledge to guide manual variable selection.

Effective variable selection for causal effect estimation involves consideration of dependencies between $\bX$ with the treatment status $T\in\{0,1\}$ and outcome $Y$. 
Let $\Ascrpitru \subseteq \{ 1,2,\ldots,p\}$ index the subset of $\bX$ upon which the PS $\pi_1(\bx)=\P(T=1\mid\bX=\bx)$ depends, 
and let $\Ascrmutru$ be an analogous index set for $\bX$ upon which either $\mu_1(\bx)$ or $\mu_0(\bx)$ depends, where $\mu_k(\bx) = \E(Y\mid \bX=\bx,T=k)$.
For any index set $\Sscr \subseteq \{ 1,2,\ldots,p\}$, let $\Sscr^c$ denote its complement in $\{ 1,2,\ldots,p\}$.  
When $\bX$ is sufficient for no unmeasured confounding, the covariates indexed in $\Ascrpitru$ is a reduced set of covariates that is also sufficient for no unmeasured confounding \citep{de2011covariate}.  However, additionally adjusting for purely prognostic covariates in $\Ascrpitru^c \cap\Ascrmutru$ can improve the efficiency of PS, outcome regression, and DR estimators \citep{lunceford2004stratification,hahn2004functional,brookhart2006variable}.

To exploit this phenomenon, we consider an inverse probability weighting (IPW) estimator where the PS is initially estimated by regularized regression. Since variable selection procedures for the PS model would select out covariates in $\Ascrpitru^c \cap \Ascrmutru$, we also estimate a regularized regression model for $\mu_k(\bx)$, for $k=0,1$, to recover variation from covariates in $\Ascrpitru^c\cap\Ascrmutru$ to inform estimation of a calibrated PS.  The calibration is implemented through smoothing $T$ over the linear predictors for $\bX$ from both the initial PS and outcome models, which can be viewed as smoothing over working propensity and prognostic scores \citep{hansen2008prognostic}.  
The resulting IPW estimator maintains double-robustness and achieves the semiparametric efficiency bound when $p$ is fixed, under correctly specified PS and outcome working models.  To the best of our knowledge, this is the first proposal in the literature that demonstrates these properties can be achieved through weighting only, without explicit augmentation.  We show that the estimator is asymptotically linear and use this to characterize large-sample robustness and efficiency properties.  The smoothing results in a refinement of the influence function under misspecification of the outcome model that can potentially result in substantial gains in efficiency relative to traditional DR estimators, which is confirmed in simulations.  
These properties hold in settings where $p$ is either fixed or allowed to diverge slowly with $n$ assuming fixed sparsity indices.

Data-driven variable selection for causal effect estimation has been considered in screening methods based on marginal associations between $\bX$ with $T$ and $Y$ \citep{schneeweiss2009high}, but the results can be misleading because marginal associations need not agree with conditional associations. \cite{de2011covariate} carefully characterized and proposed algorithms to identify minimal subsets of covariates that are sufficient for no unmeasured confounding.
Recent works have considered using regularized regression to select variables and post-selection methods that estimate treatment effects through partially linear models \citep{belloni2013inference} and DR estimators \citep{farrell2015robust,belloni2017program}. These methods focus on delivering uniformly valid inference under high-dimensional regimes assuming approximately sparse models. 
Others have proposed modifying the regularization penalty itself in a way to select the relevant covariates and estimate treatment effects through 
IPW \citep{shortreed2017outcome} and DR estimators \citep{koch2017covariate}.  However, these papers generally do not fully work out the full asymptotic distribution of the final estimator, making efficiency comparisons with established methods difficult.  Some of the methods are also only singly-robust.  
Bayesian model averaging (\cite{cefalu2017model} and references therein) offers a principled alternative for variable selection
but encounters burdensome computations that are possibly infeasible for large $p$.

Our proposed double-index PS (DiPS) can be viewed as a simple and intuitive approach to dimension reduction of $\bX$ for estimating the PS.  
The approach for DiPS closely resembles a method proposed for estimating mean outcomes in the presence of data missing at random \citep{hu2012semiparametric}, except we use the double-score to estimate a PS instead of an outcome model.  
In contrast to their results, we show that a higher-order kernel is required due to the two-dimensional smoothing, find explicit efficiency gains under misspecification of the outcome model, and consider $p$ diverging with $n$.  
There is also some similar intuition shared with collaborative DR methods \citep{van2010collaborative} in that associations with both treatment and outcome are taken into account when estimating a PS. However, DiPS takes a much different approach to estimating the PS. 
In the following, we introduce the proposed method and consider its asymptotic properties in Sections \ref{s:method} and \ref{s:asymptotics}. A perturbation-resampling method is proposed for inference in Section \ref{s:perturbation}. Simulations and applications to estimating treatment effects in an EMR study and cohort study are presented in Section \ref{s:numerical}. We conclude with some additional remarks in Section \ref{s:discussion}.

\section{Method} \label{s:method}
\subsection{Notations and Problem Setup}
Let $\bZ_i = (Y_i,T_i,\bX_i\trans)\trans$ be the observed data for the $i$th subject, where $Y_i$ is an outcome that could be modeled by a generalized linear model (GLM), $T_i\in\{ 0,1\}$ a binary treatment, and $\bX_i$ is a $p$-dimensional vector of covariates 
with support $\Xscr\subseteq\mathbb{R}^{p}$.  Here $p$ is allowed to diverge slowly with $n$ such that $log(p)/log(n) \to \nu$, for $\nu \in [0,1)$, which includes the case where $p$ is fixed by taking $\nu=0$.  For a given $n$, the observed data consists of independent and identically distributed (iid) observations $\Dscr=\{ \bZ_i : i=1,\ldots,n\}$ drawn from a distribution $\P_n$, which potentially may vary with $n$.  We suppress the dependence in the notations, 
implicitly assuming statements involving $\P$ and associated statistical functionals hold for each $n$.  Let $Y_i^{(1)}$ and $Y_i^{(0)}$ denote the counterfactual outcomes had a subject received treatment or control.  Based on $\Dscr$, we want to make inferences about the average treatment effect (ATE):
\begin{align}
\Delta = \E\{ Y^{(1)}\} - \E\{ Y^{(0)}\} = \mu_1 - \mu_0.
\end{align}

For identifiability, we require the following standard causal inference assumptions:
\begin{align}
&Y = T Y^{(1)} + (1-T) Y^{(0)} \text{ with probability } 1\label{a:consi}\\
&\pi_1(\bx) \in [\epsilon_{\pi},1-\epsilon_{\pi}] \text{ for some } \epsilon_{\pi} >0, \text{ when } \bx \in \Xscr \label{a:posi}\\
&Y^{(1)} \indep T \mid \bX \text{ and } Y^{(0)}\indep T \mid \bX, \label{a:nuca}
\end{align}
where $\pi_k(\bx) = \P(T=k\mid \bX=\bx)$, for $k=0,1$.  The third condition assumes that $\bX$ is a sufficient set of covariates such that no unmeasured confounding holds given the entire $\bX$.
Under these assumptions, $\Delta$ can be identified from the observed data distribution $\P$ through:
\begin{align*}
\Deltstr = \E\{ \mu_1(\bX) - \mu_0(\bX)\} = \E\left\{ \frac{I(T=1)Y}{\pi_1(\bX)}-\frac{I(T=0)Y}{\pi_0(\bX)}\right\},
\end{align*}
where $\mu_k(\bx) = \E(Y\mid \bX=\bx,T=k)$, for $k=0,1$.  We will consider an estimator based on the IPW form that will nevertheless be doubly-robust so that it is consistent under models where either $\pi_k(\bx)$ or $\mu_k(\bx)$ is correctly specified.

\subsection{Parametric Models for Nuisance Functions} \label{s:models}
We consider parametric modeling as a means to reduce the dimensions of $\bX$ when estimating the PS.  For reference, let $\Mscrnp$ be the nonparametric model for the distribution of $\bZ$, $\P$, that has no restrictions on $\P$ except requiring the second moment of $\bZ$ to be finite.  Let $\Mscrpi \subseteq \Mscrnp$ and $\Mscrmu \subseteq \Mscrnp$ respectively denote parametric working models under which:
\begin{align}
&\pi_1(\bx) = g_{\pi}(\alpha_0+\balph\trans\bx), \label{e:psmod} \\
\text{and } &\mu_k(\bx) = g_{\mu}(\beta_0 + \beta_1 k + \bbeta_k\trans\bx), \text{ for } k=0,1,\label{e:ormod}
\end{align}
where $g_{\pi}(\cdot)$ and $g_{\mu}(\cdot)$ are known link functions, and $\balphvec=(\alpha_0,\balph\trans)\trans\in\Thetpi \subseteq \mathbb{R}^{p+1}$ and $\bbetavec = (\beta_0,\beta_1,\bbeta_0\trans,\bbeta_1\trans)\trans \in \Thetmu \subseteq \mathbb{R}^{2p+2}$ are unknown parameters.
In \eqref{e:ormod} slopes are allowed to differ by treatment arms to allow for heterogeneous effects of $T$ for subjects with different $\bX$ even with a linear link.  When it is reasonable to assume heterogeneity is weak or nonexistent, it may be beneficial for efficiency to restrict $\bbeta_0 = \bbeta_1$.

Regardless of the validity of either working model (i.e. whether $\P\in \Mscr_{\pi} \cup \Mscr_{\mu}$),
we first obtain estimates of $\balph$ and $\bbeta_k$'s through adaptive LASSO \citep{zou2006adaptive}:
\begin{align}
&(\alphhat_0,\balphhat\trans)\trans = \argmax_{\balphvec}\left\{n^{-1}\sum_{i=1}^n \ell_{\pi}(\balphvec;T_i,\bX_i) - \lambda_{\pi,n}\sum_{j=1}^p \abs{\alpha_j} / \abs{\alphtld_{j}}^{\gamma}\right\} \label{e:psest}\\ 
&(\betahat_0,\betahat_1,\bbetahat_0\trans,\bbetahat_1\trans)\trans = \argmax_{\bbetavec}\left\{ n^{-1}\sum_{i=1}^n \ell_{\mu}(\bbetavec;\bZ_i) -\lambda_{\mu,n} \left( \abs{\beta}_1 / \abs{\betatld_1}^{\gamma}  + \sum_{k=0}^1\sum_{j=2}^p \abs{\beta_{k,j}} / \abs{\betatld_{k,j}}^{\gamma}\right) \right\}, \label{e:orest}
\end{align}
where $\ell_{\pi}(\balphvec;T_i,\bX_i)$ denotes the log-likelihood for $\balphvec$ under $\Mscr_{\pi}$ given $T_i$ and $\bX_i$, $\ell_{\mu}(\bbetavec;\bZ_i)$ is a log-likelihood for $\bbetavec$ from a GLM suitable for the outcome type of $Y$ under $\Mscr_{\mu}$ given $\bZ_i$, $\alphtld_{j}$, $\betatld_1$, and $\betatld_{k,j}$ are initial root-$n$ consistent estimates,
$\lambda_{\pi,n}$ is a tuning parmaeter such that $n^{1/2}\lambda_{\pi,n}\to 0$ and $n^{(1-\nu)(1+\gamma)/2}\lambda_{\pi,n} \to\infty$, with $\gamma > 2\nu/(1-\nu)$, and similarly for $\lambda_{\mu,n}$ \citep{zou2009adaptive}. We specify adaptive LASSO here to estimate the nuisance parameters for concreteness, but use of other penalized likelihood methods can also be justified, so long as they have an oracle property, as in Theorem 2 of \cite{zou2006adaptive} and described below.

Under model \eqref{e:psmod} and \eqref{e:ormod}, we assume that $\balph$ and $\bbeta_k$, for $k=0,1$, are sparse.  More generally, regardless of whether working models are correct or misspecified, we assume that there exist least false parameters $(\alphbar_0,\balphbar\trans)\trans$ and $(\betabar_0,\betabar_1,\bbetabar_0\trans,\bbetabar_1\trans)\trans$ \citep{lu2012robustness} such that:
\begin{align}\label{a:leastfalse}
\begin{split}
&(\alphbar_0,\balphbar\trans)\trans  \text{ uniquely maximize } \E\left\{ \ell_{\pi}(\balphvec;T_i,\bX_i)\right\} \\ &(\betabar_0,\betabar_1,\bbetabar_0\trans,\bbetabar_1\trans)\trans \text{ uniquely maximize } \E\left\{ \ell_{\mu}(\bbetavec;\bZ_i)\right\}.
\end{split}
\end{align}
Let $\Ascrpi$ and $\Ascrmuk$ be respective supports for $\balphbar$ and $\bbetabar_k$ and let $s_{\balph}= \abs{\Ascrpi}$ and $s_{\bbeta_k} = \abs{\Ascrmuk}$ be the sparsity indices.   We further assume $\balphbar$ and $\bbetabar_k$ have fixed sparsity such that:
\begin{align} \label{a:spars}
s_{\balph}, s_{\bbeta_0} \text{ and } s_{\bbeta_1} \text{ are fixed as } n\to\infty .
\end{align}
For any vector $\bv$ of length $p$ and any index set $\Sscr \subseteq \{ 1,2,\ldots,p\}$, let $\bv_{\Sscr}$ denote the subvector of $\bv$ restricted to elements indexed in $\Sscr$. Assumption \eqref{a:leastfalse} is a high-level assumption that would be required for $\balphhat$ and $\bbetahat_k$ to maintain an oracle property with respect to the least false parameters $\balphbar$ and $\bbetabar_k$ under possibly misspecified working models.  Under this assumption using arguments similar to those in \cite{lu2012robustness} and \cite{zou2009adaptive} it can be shown that $\P(\balphhat_{\Ascrpi^c}=\bzero)\to 1$ and admits an expansion of the form $n^{1/2}(\balphhat-\balphbar)_{\Ascrpi}=n^{-1/2}\sum_{i=1}^n \bPsi_{i,\Ascrpi} + o_p(1)$, which would yield the asymptotic normality results of the oracle property,
and similarly for $\bbetahat_k$.  We rely on these results along with \eqref{a:spars} to show that the DiPS IPW is asymptotically linear in Theorem \ref{t:IFexp}. In regimes where $\nu>0$, \eqref{a:spars} models a setting in which a small number of covariates exhibit non-negligible associations with $T$ and $Y$ and a majority of covariates are noise.  Assumption \ref{a:spars} may not be required for asymptotic linearity and can potentially be relaxed allowing $s_{\balph}$ and $s_{\bbeta_k}$ to diverge slowly, for example, if they are $o(n^{1/3})$. We invoke this assumption to avoid complications of a growing support, which may need triangular array asymptotics to accommodate dependence of the support on $n$.

\subsection{Double-Index Propensity Score and IPW Estimator} \label{s:estimator}
To mitigate the effects of misspecification of \eqref{e:psmod}, one could perform nonparametric smoothing of $T$ over $\balphhat\trans\bX$ to calibrate the initial PS estimator $g_{\pi}(\alphhat_0 + \bX\trans\balphhat)$.
We consider smoothing over not only $\balphhat\trans\bX$ but also $\bbetahat_k\trans\bX$ as well to allow variation in prognostic covariates indexed in $\Ascrmuk$ to inform this calibration.  Such covariates are reduced into $\bbetahat_k\trans\bX$ to allow for nonparametric kernel smoothing in low (two) dimensions.
The DiPS estimator for each treatment is:
\begin{align}
\pihat_k(\bx;\bthethat_k)= \frac{n^{-1}\sum_{j=1}^nK_h\{(\balphhat,\bbetahat_k)\trans(\bX_j-\bx)\}I(T_j=k)}{n^{-1}\sum_{j=1}^n K_h\{(\balphhat,\bbetahat_k)\trans(\bX_j-\bx)\}}, \text{ for } k=0,1,
\end{align}
where $\bthethat_k = (\balphhat\trans,\bbetahat_k\trans)\trans$, $K_h(\bu)= h^{-2}K(\bu/h)$, and $K(\bu)$ is a bivariate $q$-th order kernel function with $q>2$.  A higher-order kernel is required here for the asymptotics to be well-behaved, which is the price for estimating the nuisance functions $\pi_k(\bx)$ using two-dimensional smoothing.  This allows for the possibility of negative values for $\pihat_k(\bx;\bthethat_k)$.  Nevertheless, $\pihat_k(\bx;\bthethat_k)$ are nuisance estimates not of direct interest, and we find that such negative PS estimates typically occur infrequently, occurring on average in simulations in $0.01\%$ to $2.10\%$ of observations depending the size of $n$ and $p$ across scenarios where working models are correct or incorrectly specified (Web Appendix D).  
As they are infrequent and do not appear to compromise the performance of the final estimator, they can potentially be left as is when encountered in practice. Alternatively, methods that discard or trim PS estimates to handle near-violations of positivity, as in Assumption \eqref{a:posi}, can be considered \citep{crump2009dealing}.
A monotone transformation of the input scores for each treatment $\bShat_k = (\balphhat,\bbetahat_k)\trans\bX$ can be applied prior to smoothing to improve finite sample performance \citep{wand1991transformations}.  In numerical studies, for instance, we applied a probability integral transform based on the normal cumulative distribution function to the standardized scores to obtain approximately uniformly distributed inputs.  The components of $\bShat_k$ can also be scaled such that a common bandwidth $h$ can be used for both components of the score.

With $\pi_k(\bx)$ estimated by $\pihat_k(\bx;\bthethat_k)$, the estimator for $\Delta$ is given by $\Delthat = \muhat_1 - \muhat_0$, where:
\begin{align}
\muhat_k = \left\{ \sum_{i=1}^n \frac{I(T_i = k)}{\pihat_k(\bX_i;\bthethat_k)}\right\}^{-1}\left\{ \sum_{i=1}^n \frac{I(T_i = k)Y_i}{\pihat_k(\bX_i;\bthethat_k)}\right\}^{-1}, \text{ for } k=0,1.
\end{align}
This is the usual normalized IPW estimator, where the PS is estimated by the DiPS.  The intuition for double-robustness of the estimator is as follows.  Regardless of the validity of either working model, provided the asymptotics are well-behaved, $\muhat_k$ is consistent for:
\begin{align*}
\mubar_k = \E\left\{ \frac{I(T_i=k)Y_i}{\pi_k(\bX_i;\bthetbar_k)}\right\}, \text{ for } k=0,1,
\end{align*}
where $\bthetbar_k = (\balphbar\trans,\bbetabar_k\trans)\trans$, and $\pi_k(\bx;\bthetbar_k) =\P(T_i=k \mid \balphbar\trans\bX_i=\balphbar\trans\bx,\bbetabar_k\trans\bX_i = \bbetabar_k\trans\bx)$.  Under $\Mscrpi$, $\pi_k(\bx;\bthetbar_k) = \pi_k(\bx)$ so that the estimand, under the causal assumptions \eqref{a:consi}-\eqref{a:nuca}, reduces to:
\begin{align*}
\mubar_k = \E\left\{ \frac{I(T_i=k)Y_i}{\pi_k(\bX_i)}\right\} = \E \left\{ Y_i^{(k)}\right\}, \text{ for } k=0,1.
\end{align*}
On the other hand, under $\Mscrmu$, $\E(Y_i \mid \balphbar\trans\bX_i=\balphbar\trans\bx,\bbetabar_k\trans\bX_i=\bbetabar_k\trans\bx,T_i =k) = \mu_k(\bx)$ so that:
\begin{align*}
\mubar_k &= \E\left\{ \E(Y_i \mid \balphbar\trans\bX_i,\bbetabar\trans\bX_i,T_i =k)\right\} =\E\left\{ \mu_k(\bX_i)\right\} = \E\{ Y_i^{(k)}\}, \text{ for } k=0,1.
\end{align*}
In the following, we show that $\muhat_k$ (and thus $\Delthat$) are asymptotically linear. We then subsequently examine robustness and efficiency properties using the expansion.

\section{Asymptotic Robustness and Efficiency Properties} \label{s:asymptotics}
We directly show in Web Appendix B that $\muhat_k$ is asymptotically linear for $k=0,1$ in general without assuming either of the working models are correct.  Let $\Deltbar = \mubar_1 - \mubar_0$ and 
$\Wscrhat_k = n^{1/2}(\muhat_k - \mubar_k)$ for $k=0,1$ so that $n^{1/2}(\Delthat - \Deltbar) = \Wscrhat_1 - \Wscrhat_0$.
\begin{theorem} \label{t:IFexp}
Suppose that causal assumptions \eqref{a:consi}-\eqref{a:nuca}, the least false parameter and sparsity assumptions \eqref{a:leastfalse}-\eqref{a:spars} and regularity conditions in Web Appendix A hold.  
If $log(p)/log(n) \to \nu$ for $\nu \in [0,1)$, then $\muhat_k$ is asymptotically linear in that it admits the expansion:
\begin{align}
\Wscrhat_k &= n^{-1/2}\sum_{i=1}^n \frac{I(T_i =k)Y_i}{\pi_k(\bX_i;\bthetbar_k)} - \left\{ \frac{I(T_i = k)}{\pi_k(\bX_i;\bthetbar_k)}-1\right\}\E(Y_i \mid \balphbar\trans\bX_i,\bbetabar_k\trans\bX_i,T_i =k ) -\mubar_k \label{e:IFeff}\\
&\qquad + n^{-1/2}\sum_{i=1}^n \bu_{k,\Ascrpi}\trans \bPsi_{i,\Ascrpi} + \bv_{k,\Ascrmuk}\trans \bUpsi_{i,k,\Ascrmuk} + O_p(n^{1/2}h^q + n^{-1/2}h^{-2}), \label{e:IFnui}
\end{align}
for $k=0,1$, where $\bu_{k,\Ascrpi}$ and $\bu_{k,\Ascrmuk}$ are deterministic vectors, $\bPsi_{i,\Ascrpi}$ and $\bUpsi_{i,k,\Ascrmuk}$ are influence functions from asymptotic expansions of $\balphhat_{\Ascrpi}$ and $\bbetahat_{k,\Ascrmuk}$.
Under model $\Mscrpi$ $\bv_{k,\Ascrmuk} = \bzero$, for $k=0,1$.  Under $\Mscrpi \cap \Mscrmu$, we additionally have that $\bu_{k,\Ascrpi} = \bzero$, for $k=0,1$.
\end{theorem}
\begin{hproof} $\Wscrhat_k$ can be decomposed as:
\begin{align*}
\Wscrhat_k &= n^{-1/2}\sum_{i=1}^n \frac{I(T_i=k)}{\pi_k(\bX_i;\bthetbar_k)} (Y_i - \mubar_k) + n^{-1/2}\sum_{i=1}^n \left\{ \frac{I(T_i=k)}{\pihat_k(\bX_i;\bthetbar_k)} - \frac{I(T_i=k)}{\pi_k(\bX_i;\bthetbar_k)}\right\} (Y_i - \mubar_k) \\
&\qquad + n^{-1/2}\sum_{i=1}^n \left\{ \frac{I(T_i=k)}{\pihat_k(\bX_i;\bthethat_k)} - \frac{I(T_i=k)}{\pi_k(\bX_i;\bthetbar_k)}\right\} (Y_i - \mubar_k) + o_p(1).
\end{align*}
The first term directly contributes to the expansion. The second term is the contribution from re-estimating the PS through kernel smoothing given $\bthetbar_k$.  
We apply a V-statistic projection lemma \citep{newey1994large} to obtain an asymptotically linear representation.  The third term can be expanded by Taylor expansion into terms of the form $\bu_k\trans n^{1/2}(\balphhat-\balphbar)$ and $\bv_k\trans n^{1/2}(\bbetahat-\bbetabar)$.  Applying the selection consistency that $\P(\balphhat_{\Ascrpi^c} = \bzero) \to 1$,
$\bu_k\trans n^{1/2}(\balphhat-\balphbar) = \bu_{k,\Ascrpi}\trans n^{1/2}(\balphhat-\balphbar)_{\Ascrpi} + o_p(1)$.
Lastly, we use that $n^{1/2}(\balphhat-\balphbar)_{\Ascrpi}=n^{-1/2}\sum_{i=1}^n \bPsi_{i,\Ascrpi} + o_p(1)$ and work out the forms of the loading vector $\bu_{k,\Ascrpi}$ and repeat for $\bbetahat_k$ to complete the expansion.
\end{hproof}

Let $\Delthatdr=\muhat_{1,dr}-\muhat_{0,dr}$ denote the usual doubly-robust estimator, as in Equation (9) of \cite{lunceford2004stratification}, with the PS $\pi_k(\bx)$ and mean outcome $\mu_k(\bx)$ estimated in the same way as through \eqref{e:psest} and \eqref{e:orest}.  The influence function expansion for $\Delthat$ in Theorem \ref{t:IFexp} is nearly identical to that of $\Delthatdr$.  The terms in \eqref{e:IFeff} would be the same except $\pi_k(\bX_i;\bthetbar_k)$ and $\E(Y_i \mid \balphbar\trans\bX_i,\bbetabar_k\trans\bX_i,T_i=k)$ replaces asymptotic estimates under parametric models.  Terms in \eqref{e:IFnui} analogously represent the additional contributions from estimating the nuisance parameters.  No contribution from smoothing is incurred provided the bandwidths are suitably chosen.  This similarity in the influence functions yields similar robustness and efficiency properties, which are improved upon under model misspecification due to the smoothing.

\subsection{Robustness}
As a consequence of Theorem \ref{t:IFexp}, $\Delthat$ is root-$n$ consistent for $\Deltbar$ so that $\Delthat - \Deltbar = O_p(n^{-1/2})$ provided that $h=O(n^{-\alpha})$ for $\alpha \in (\frac{1}{2q},\frac{1}{4})$.  As discussed in Section \ref{s:estimator}, under $\Mscrpi \cup \Mscrmu$, $\Deltbar=\Delta$.  Hence $\Delthat$ is \emph{doubly-robust} for $\Delta$ in that $\Delthat$ is root-$n$ consistent for $\Delta$ under $\Mscrpi\cup\Mscrmu$.
Beyond this usual form of double-robustness, if the PS model specification is incorrect, we expect the calibration step to at least partially correct for the misspecfication in large samples since $\pi_k(\bx;\bthetbar_k)$ is closer to the true $\pi_k(\bx)$ than the misspecified parametric model $g_{\pi}(\alphbar_0+\balphbar\trans\bx)$.  
Let $\widetilde{\Mscr}_{\pi}$ denote a model under which $
\pi_1(\bx) = \widetilde{g}_{\pi}(\balph\trans\bx)$ for some \emph{unknown} link function $\widetilde{g}_{\pi}(\cdot)$ and unknown $\balph \in\mathbb{R}^{p}$, \emph{and} $\bX$ are known to be elliptically distributed such that $\E(\ba\trans\bX\mid \balph_*\trans\bX)$ exists and is linear in $\balph_*\trans\bX$, where $\balph_*$ denotes the true $\balph$ (e.g. if $\bX$ is multivariate normal).  By the results of \cite{li1989regression}, it can be shown that $\balphbar=c\balph_*$ for some scalar $c$ under $\Mscrpitld$.  But since $\pihat_k(\bx;\bthethat_k)$ is consistent for  $\pi_k(\bx;\bthetbar_k)=\P(T=k\mid\balphbar\trans\bX=\balphbar\trans\bx,\bbetabar_k\trans\bX=\bbetabar_k\trans\bx)$, it recovers $\pi_k(\bx)$ under $\Mscrpitld$.  Consequently, $\Delthat$ also has some mild benefits in robustness in that $\Delthat-\Delta = O_p(n^{-1/2})$ under the slightly larger model $\Mscrpi\cup \Mscrpitld\cup\Mscrmu$.
The same phenomenon also occurs when estimating $\bbeta_k$ under misspecification of the link in \eqref{e:ormod}, if we do not assume $\bbeta_0=\bbeta_1$.
In this case, if $\Mscrmutld$ is an analogous model under which $\mu_1(\bx) = \widetilde{g}_{\mu,1}(\bbeta_1\trans\bx)$ and $\mu_0(\bx)=\widetilde{g}_{\mu,0}(\bbeta_0\trans\bx)$
for some unknown link functions $\widetilde{g}_{\mu,0}(\cdot)$ and $\widetilde{g}_{\mu,1}(\cdot)$ and $\bX$ are elliptically distributed, then $\Delthat-\Delta = O_p(n^{-1/2})$ under the slightly larger model $\Mscrpi\cup \Mscrpitld \cup \Mscrmu\cup \Mscrmutld$.  This does not hold when $\bbeta_0=\bbeta_1$, as $T$ is binary so $(T,\bX\trans)\trans$ is not exactly elliptically distributed.  But the result may still be expected to hold approximately.

\subsection{Efficiency} \label{ss:eff}
Let the terms contributed to the influence function for $\Delthat$ when $\balph$ and $\bbeta_k$ are known be:
\begin{align}
\varphi_{i,k} = \frac{I(T_i =k)Y_i}{\pi_k(\bX_i;\bthetbar_k)} - \left\{ \frac{I(T_i = k)}{\pi_k(\bX_i;\bthetbar_k)}-1\right\}\E(Y_i \mid \balphbar\trans\bX_i,\bbetabar_k\trans\bX_i,T_i =k ) -\mubar_k.
\end{align}
Under $\Mscrpi \cap \Mscrmu$, $\varphi_{i,k}$ is the full influence function for $\Delthat$.  This is the efficient influence function for $\Deltstr$ under $\Mscrnp$ at distributions for $\P$ belonging to $\Mscrpi \cap \Mscrmu$ when $p$ is fixed \citep{robins1994estimation,tsiatis2007semiparametric}, since $\E(Y_i \mid \balphbar\trans\bX_i=\balphbar\trans\bx,\bbetabar_k\trans\bX_i=\bbetabar_k\trans\bx,T_i=k)=\mu_k(\bx)$ and $\pi_k(\bx;\bthetbar_k)=\pi_k(\bx)$.  When $\nu>0$ so that $p$ diverges with $n$, there are no well-established semiparametric efficiency bounds.  However with fixed sparsity indices \eqref{a:spars}, the asymptotic variance still reaches the same bound had $p$ been fixed.

Beyond this characterization of efficiency that parallels that of $\Delthatdr$, there are additional benefits of $\Delthat$ under $\Mscrpi\cap \Mscrmu^c$.  In this case, akin to $\Delthatdr$, estimating $\bbeta_k$ does not contribute to the asymptotic variance since $\bv_{k,\Ascrmuk}=\bzero$, and a similar $n^{1/2}\bu_{k,\Ascrpi}\trans(\balphhat-\balphbar)_{\Ascrpi}$ term is contributed from estimating $\balph$.  The analogous term in the expansion for $\Delthatdr$ contributes the negative of a projection of the preceding terms onto the linear span of the score function for $\balph$, restricted to components in $\Ascrpi$, to its influence function (Section 9.1 of \cite{tsiatis2007semiparametric}).  The same interpretation of the influence function can be adopted for $\Delthat$.
\begin{theorem} \label{t:effgain}
Let $\bU_{\balph}$ be the score for $\balph$ under $\Mscrpi$ and let $[\bU_{\balph,\Ascrpi}]$ denote the linear span of its components indexed in $\Ascrpi$. In the Hilbert space of random variables with mean $0$ and finite variance $\Lscr_2^0$ with inner product given by the covariance, let $\Pi\{ V \mid \Sscr\}$ denote the projection of some $V\in \Lscr_2^0$ into a subspace $\Sscr \subseteq \Lscr_2^0$.  If the assumptions required for Theorem \ref{t:IFexp} hold 
, under $\Mscrpi$, $\bu_{k,\Ascrpi}\trans n^{1/2}(\balphhat-\balphbar)_{\Ascrpi} = -n^{-1/2}\sum_{i=1}^n\Pi\{\varphi_{i,k}\mid[\bU_{\balph,\Ascrpi}]\} + o_p(1)$.
\end{theorem}
The proof is based on simplifying $\bu_{k,\Ascrpi}$ and is given in Web Appendix B. This result can be used to show that the asymptotic variance of $\Delthat$ is lower than that of $\Delthat_{dr}$ under $\Mscrpi \cap \Mscrmu^c$.
Based on this result, under $\Mscrpi \cap \Mscrmu^c$ the influence function for $\muhat_k$ is $\varphi_{i,k}-\Pi\left\{ \varphi_{i,k} \mid [\bU_{\balph,\Ascrpi}]\right\}$, and for the usual DR estimator $\muhat_{k,dr}$ is $\phi_{i,k} - \Pi\left\{ \phi_{i,k} \mid [\bU_{\balph,\Ascrpi}]\right\}$, where:
\begin{align*}
\phi_{i,k} = \frac{I(T_i =k)Y_i}{\pi_k(\bX_i)} - \left\{ \frac{I(T_i = k)}{\pi_k(\bX_i)}-1\right\}g_{\mu}(\betabar_{0} + \betabar_1 k+\bbetabar_k\trans\bX_i) -\mubar_{k}.
\end{align*}
But since $\E(Y_i\mid \balphbar\trans\bX_i=\balphbar\trans\bx,\bbetabar_k\trans\bX_i=\bbetabar_k\trans\bx,T_i =k)$ better approximates $\mu_k(\bx)$ than the asymptotic estimate under the misspecified parametric model $g_{\mu}(\betabar_0 + \betabar_1 k +\bbetabar_k\trans\bx)$, it can then be shown that $\E(\phi_{i,k}^2)>\E(\varphi_{i,k}^2)$ for $k=0,1$. Since the influence functions involve projections onto the same space $[\bU_{\balph,\Ascrpi}]$, it can be seen through geometric argument that $\E\left[ \varphi_{i,k}-\Pi\left\{ \varphi_{i,k} \mid [\bU_{\balph,\Ascrpi}]\right\}\right]^2 < \E\left[ \phi_{i,k} - \Pi\left\{ \phi_{i,k} \mid [\bU_{\balph,\Ascrpi}]\right\}\right]^2$, so that $\Delthat$ is more efficient than $\Delthat_{dr}$ under $\Mscrpi \cap \Mscrmu^c$.
We show in the simulation studies that this improvement can lead to substantial efficiency gains under $\Mscrpi \cap \Mscrmu^c$ in finite samples.  These unique robustness and efficiency properties distinguish $\Delthat$ from $\Delthat_{dr}$ and its variants.
We next consider a perturbation scheme to estimate standard errors (SE) and confidence intervals (CI) for $\Delthat$.

\section{Perturbation Resampling} \label{s:perturbation}
Although the asymptotic variance of $\Delthat$ can be determined through its influence function specified in Theorem \eqref{t:IFexp}, a direct empirical estimate based on the influence function is infeasible because it involves functionals of $\P$ that are difficult to estimate.  Instead we propose a simple perturbation-resampling procedure.  Let $\Gscr = \{ G_i : i=1,\ldots,n\}$ be a set of non-negative iid random variables with unit mean and variance independent of $\Dscr$.  The procedure perturbs each ``layer'' of the estimation of $\Delthat$.  Let the perturbed estimates of $\balphvec$ and $\bbetavec$ be:
\begin{align*}
&(\alphhat^*_0,\balphhat^{*\sf \tiny T})\trans = \argmax_{\balphvec}\left\{ n^{-1}\sum_{i=1}^n \ell_{\pi}(\balphvec;T_i,\bX_i)G_i - \lambda_{\pi,n}\sum_{j=1}^p \abs{\alpha_j}/\abs{\alphtld_j^*}^{\gamma}\right\} \\
&(\betahat^*_0,\betahat^*_1,\bbetahat^{*\sf \tiny T}_0,\bbetahat^{*\sf \tiny T}_1)\trans = \argmax_{\bbetavec}\left\{ n^{-1}\sum_{i=1}^n \ell_{\mu}(\bbetavec;\bZ_i)G_i -\lambda_{\mu,n}\sum_{j=1}^p \abs{\beta_j}/\abs{\betatld_j^*}^{\gamma}\right\},
\end{align*}
where $\alphtld_j^*$ and $\betatld_j^*$ are perturbed initial estimates obtained from analogously perturbing its estimating equations.
The perturbed DiPS estimates are calculated by:
\begin{align*}
\pihat_k^*(\bx;\bthethat_k^*)= \frac{\sum_{j=1}^nK_h\{(\balphhat^*,\bbetahat_k^*)\trans(\bX_j-\bx)\}I(T_j=k)G_j}{\sum_{j=1}^n K_h\{(\balphhat^*,\bbetahat_k^*)\trans(\bX_j-\bx)\}G_j}, \text{ for } k=0,1.
\end{align*}
Lastly the perturbed estimator is given by $\Delthat^* = \muhat_1^* - \muhat_0^*$ where:
\begin{align*}
\muhat_k^* = \left\{ \sum_{i=1}^n \frac{I(T_i = k)}{\pihat_k^*(\bX_i;\bthethat_k^*)}G_i\right\}^{-1}\left\{ \sum_{i=1}^n \frac{I(T_i = k)Y_i}{\pihat_k^*(\bX_i;\bthethat_k^*)} G_i\right\}^{-1}, \text{ for } k=0,1.
\end{align*}
It can be shown based on arguments in \cite{jin2001simple} that the asymptotic distribution of $n^{1/2}(\Delthat - \Deltbar)$ coincides with that of $n^{1/2}(\Delthat^* - \Delthat) \mid \Dscr$.  We can thus approximate the SE of $\Delthat$ based on the empirical standard deviation or, as a robust alternative, the mean absolute deviations (MAD) of resamples $\Delthat^*$ and construct CI's using percentiles of resamples.

\section{Numerical Studies} \label{s:numerical}
\subsection{Simulation Study}
We performed extensive simulations to assess the finite sample bias and relative efficiency (RE) of $\Delthat$ (DiPS) compared to alternative estimators.  We also assessed the performance of the perturbation procedure.  
Throughout in implementing the adaptive LASSO, we used ridge regression for the initial estimators $\alphtld_j$ and $\betatld_j$ where the ridge tuning parameter chosen by minimizing the Akaike information criterion (AIC).  The adaptive LASSO tuning parameter was chosen by an extended regularized information criterion \citep{hui2015tuning}, which exhibited relatively good performance for variable selection.  We refitted models with selected covariates to reduce bias, as suggested in \cite{hui2015tuning}. The power parameter $\gamma$ was set as $\lceil \frac{2\nu}{1-\nu}\rceil+ 1$, where $\nu = log(p)/log(n)$.
A Gaussian product kernel of order $q=4$ with a plug-in bandwidth at the optimal order (see \hyperref[s:discussion]{Discussion}) was used for smoothing.  For comparison, we considered alternative standard estimators with nuisances estimated by regularization and recently developed methods for estimating ATE that incorporate variable selection: (1) IPW with $\pi_1(\bx)$ estimated by adaptive LASSO (ALAS),
(2) $\Delthat_{dr}$ with nuisances estimated by adaptive LASSO (DR-ALAS),
(3) Modification of $\Delthat_{dr}$ in which $\pi_1(\bx)$ and $\mu_k(\bx)$ are estimated by separate one-dimensional kernel smoothing of $T \sim \balphhat\trans\bX$ and $Y \sim \bbetahat_k\trans\bX$ among those assigned to $T=k$, for $k=0,1$ (DR-SIM), to allow for estimation of single index models (SIM) for $\pi_1(\bx)$ and $\mu_k(\bx)$,
(4) Outcome-adaptive LASSO (OAL) \citep{shortreed2017outcome},
(5) Group Lasso and Doubly Robust Estimation (GliDeR) \citep{koch2017covariate},
(6) Model averaged doubly-robust estimator (MADR) \citep{cefalu2017model}.
OAL and GLiDeR were implemented with default settings from code provided in the Supplementary Materials of the respective papers.
MADR was implemented using the \texttt{madr} package with $M=500$ Markov chain Monte Carlo (MCMC) iterations to reduce the computations.  Throughout the numerical studies, we specified $g_{\pi}(u)=1/(1+e^{-u})$ for $\Mscrpi$ and $g_{\mu}(u)=u$ with $\bbeta_0=\bbeta_1$ for $\Mscrmu$ as the working models.

The covariates were generated to approximate the distribution of the covariates from the statins EMR data from Section \ref{s:statin}. This was done to allow for non-elliptically distributed covariates that mimic the distribution of a real dataset.  Initially we generated $\bXtld\sim N(\bmutld,\bSigtld)$ where $\bmutld$ and $\bSigtld$ were the empirical mean and covariance matrix of the 15 covariates, which included $9$ binary, $3$ continuous, and $3$ log-transformed count variables. 
For binary variables we thresholded the corresponding components of $\bXtld$ so that its mean matched those in $\bmutld$, as in 
$I\{\widetilde{\sigma}_j^{-1}(\widetilde{X}_j-\widetilde{\mu}_{j}) > \Phi^{-1}(1-\widetilde{\mu}_j)\}$, where $\widetilde{\sigma}_j^2$ and $\mutld_j$ are the empirical variance and mean of the $j$-th covariate and $\Phi(\cdot)$ is the standard normal cumulative distribution function (CDF).  Lastly, we centered and standardized to obtain the final covariates $\bX = diag(\bSigtld^{-1/2})(\bXtld-\bmutld)$. The pairwise correlations of $\bX$ were generally low, mostly ranging between $-.2$ and $.2$ (full correlation matrix reported in Web Appendix C). For settings with $p>15$, we generated independent groups of the $15$ covariates that maintained the correlation structure within each group.

We subsequently focused on a continuous outcome, generating the data according to
$T\mid \bX \sim Ber\{ \pi_1(\bX)\}$ and $Y\mid \bX,T \sim N\{ \mu_T(\bX), 10^2\}$.
The simulations varied over scenarios where working models were correct or misspecified in which the true $\pi_1(\bx)$ and $\mu_k(\bx)$ are:
\begin{align*}
&\text{Both correct: } \pi_1(\bx) = g_{\pi}(.2+\balph\trans\bx), \quad \mu_k(\bx) = k + \bbeta\trans\bx \\
&\text{Misspecified $\mu_k(\bx)$: } \pi_1(\bx) = g_{\pi}(.2+\balph\trans\bx), \quad \mu_k(\bx) = k + \bbeta_{[1]}\trans\bx (1 + \bbeta_{[2]}\trans\bx) + k \bzeta\trans\bx \\
&\text{Misspecified $\pi_k(\bx)$: } \pi_1(\bx) = g_{\pi}\left\{.2+\balph_{[1]}\trans\bx(1+\balph_{[2]}\trans\bx)\right\}, \quad \mu_k(\bx) = k + \bbeta\trans\bx,
\end{align*}
where the coefficients are $\balph = .01\cdot(1,2,3,4,5,6,\bzero_{3},3,7,0,7,-5,0,\bzero_{p-15})\trans$, $\balph_{[1]} = \balph, \balph_{[2]} = (.02,.06,.02,.02,-.1,.02,\bzero_3,-.14,.1,0,-.1,.14,0,\bzero_{p-15})\trans$, $\bzeta = (\bzero_6,1,\bzero_3,1,\bzero_2,1,0,\bzero_{p-15})\trans$, $\bbeta=(\mathbf{0}_3,1,.5,.25,.125,.0625,.03125,0,1,.5,0,.25,.125,\bzero_{p-15})\trans$, $\bbeta_{[1]} = (\bzero_3,.5,0,.5,\bone_3,0,1,$ $2,0,1,2,\bzero_{p-15})\trans$, $\bbeta_{[2]} = (\bzero_3,-1.5,.75,-1.5,\bzero_3,0,-1.5,-.75,0,1.5,.75,\bzero_{p-15})\trans$,
and $\mathbf{a}_{m}$ denotes a $1\times m$ vector that has all its elements as $a$.  
For the misspecified scenarios, either $\mu_k(\bx)$ or $\pi_1(\bx)$ is a double-index model that includes both linear terms in $\bx$ and quadratic and two-way interaction terms among $\bx$ that are omitted by linear working models.  In the misspecified $\mu_k(\bx)$ case, the second index $\bbeta_{[2]}\trans\bx$ has some correlation with the PS index $\balph\trans\bx$, modeling a situation in which there exist common latent factors not fully captured by a linear outcome model.  The outcome model also includes an interaction term between $\bx$ and treatment to allow for treatment effect heterogeneity.
The parameters are set such that there are $5$ covariates belonging to each of $\Ascrpitru \cap \Ascrmutru$ (i.e. confounders), $\Ascrpitru \cap \Ascrmutru^c$ (instruments), and $\Ascrpitru^c \cap \Ascrmutru$ (pure prognostic) when $p=15$. The simulations were run for $R=1,000$ repetitions.

Table \ref{tab:bias} presents the bias and root mean square error (RMSE) for $n=500, 5,000$ when $p=15$.  Among the three scenarios considered, the bias for DiPS is small relative to the RMSE and generally diminishes towards zero as $n$ increases, verifying its double-robustness.  There remains some minor bias that persists when $n=5,000$ for DiPS that is likely a result of bias from the smoothing, as DR-SIM also incurs similar residual bias.  IPW-ALAS and OAL are singly-robust and the bias does not necessary diminish under the misspecified $\pi_1(\bx)$ scenario, although their bias is also minor in the setting considered.  MADR exhibited substantial bias under misspecified $\mu_k(\bx)$ scenario that persisted in large samples, possibly due to selecting out confounders with weak outcome associations in its emphasis on selection of prognostic covariates.  The results for bias for $p=50,100$ exhibited similar patterns.

\begin{table}[htbp]
  \scalebox{1}{
    \begin{tabular}{rrcccccccc}
    \toprule
          &       & \multicolumn{2}{c}{Both Correct} & \multicolumn{2}{c}{Misspecified $\mu_k(\bx)$} & \multicolumn{2}{c}{Misspecified $\pi_1(\bx)$} \\
    \multicolumn{1}{c}{\textbf{Size}} & \multicolumn{1}{c}{\textbf{Estimator}} & \textbf{Bias} & \textbf{RMSE} & \textbf{Bias} & \textbf{RMSE} & \textbf{Bias} & \textbf{RMSE} \\ \midrule
    \multicolumn{1}{c}{\multirow{8}[2]{*}{n=500}} & \multicolumn{1}{c}{IPW-ALAS} & 0.029 & 0.350 & 0.074 & 1.754 & 0.023 & 0.294 \\
    \multicolumn{1}{c}{} & \multicolumn{1}{c}{DR-ALAS} & 0.002 & 0.330 & 0.029 & 1.684 & -0.001 & 0.285 \\
    \multicolumn{1}{c}{} & \multicolumn{1}{c}{DR-SIM} & -0.021 & 0.315 & 0.127 & 1.495 & 0.013 & 0.287 \\
    \multicolumn{1}{c}{} & \multicolumn{1}{c}{OAL} & 0.008 & 0.321 & 0.074 & 1.484 & 0.001 & 0.284 \\
    \multicolumn{1}{c}{} & \multicolumn{1}{c}{GLiDeR} & 0.001 & 0.299 & 0.087 & 1.238 & 0.006 & 0.282 \\
    \multicolumn{1}{c}{} & \multicolumn{1}{c}{MADR} & 0.022 & 0.300 & 0.172 & 1.247 & 0.008 & 0.282 \\
    \multicolumn{1}{c}{} & \multicolumn{1}{c}{DiPS} & -0.017 & 0.319 & 0.101 & 1.193 & 0.013 & 0.293 \\ \hline
    \multicolumn{1}{c}{\multirow{8}[2]{*}{n=5,000}} & \multicolumn{1}{c}{IPW-ALAS} & 0.001 & 0.111 & -0.002 & 0.588 & 0.033 & 0.108 \\
    \multicolumn{1}{c}{} & \multicolumn{1}{c}{DR-ALAS} & -0.003 & 0.106 & -0.014 & 0.564 & -0.008 & 0.089 \\
    \multicolumn{1}{c}{} & \multicolumn{1}{c}{DR-SIM} & -0.012 & 0.103 & 0.029 & 0.516 & -0.004 & 0.089 \\
    \multicolumn{1}{c}{} & \multicolumn{1}{c}{OAL} & -0.002 & 0.105 & 0.000 & 0.527 & -0.007 & 0.089 \\
    \multicolumn{1}{c}{} & \multicolumn{1}{c}{GLiDeR} & -0.001 & 0.098 & 0.034 & 0.413 & -0.006 & 0.088 \\
    \multicolumn{1}{c}{} & \multicolumn{1}{c}{MADR} & 0.000 & 0.099 & 0.124 & 0.418 & -0.008 & 0.089 \\
    \multicolumn{1}{c}{} & \multicolumn{1}{c}{DiPS} & -0.016 & 0.106 & 0.041 & 0.349 & -0.003 & 0.091 \\
    \bottomrule
    \end{tabular}
    }
\vspace{1.5em}
\caption{Bias and RMSE of estimators by $n$ and model specification scenario for $p=15$.}
\label{tab:bias}
\end{table}

Figure \ref{fig:RE} presents the RE under the different scenarios for $n=500, 5,000$ and $p=15,50,100$.  RE was defined as the ratio of the mean square error (MSE) for DR-ALAS relative to that of each estimator, with RE $>1$ indicating greater efficiency compared to DR-ALAS.  Under the ``both correct'' scenario many of the estimators generally exhibit similar efficiency, which can be expected since many are variants of the usual DR estimator and reach the semiparametric efficiency bound.  When $n=500$ and $p=60$, there are some slightly greater differences, with GliDeR and MADR leading in efficiency gains, possibly due to differences in the variable selection performance. These differences in efficiency appear to temper when sample size is increased for $n=5,000$ and $p=60$. 
The results are similar in the ``misspecified $\pi_1(\bx)$'' scenario, where most estimators exhibited similar efficiency.

In the ``misspecified $\mu_k(\bx)$'' scenario, DiPS achieves over 70\% efficiency gain compared to GliDeR and MADR and over 140\% compared to DR-SIM in the large sample setting when $n=5,000$ and $p=15$. This suggests that expected efficiency gains under misspecified outcome models due to the results of Section \ref{ss:eff} can be substantial.  Even if $\pi_1(\bx)$ and $\mu_k(\bx)$ are estimated under a SIM, there are still gains from DiPS when the PS direction $\balphbar\trans\bX$ is informative of the mean outcome beyond $\bbetabar\trans_k\bX$. These gains diminish when $p$ is larger relative to $n$, possibly due to imperfect variable selection. Again GLiDeR and MADR achieve the highest efficiency when $n=500$ and $p=60$, notwithstanding the substantial bias of MADR.  Thus the performance of DiPS using adaptive LASSO can be somewhat compromised when $p$ is very large relative to $n$ and the variable selection performance is sub-optimal.

\begin{figure}[h!]
\centering
\centerline{(a) Both correct}
\includegraphics[scale=.425]{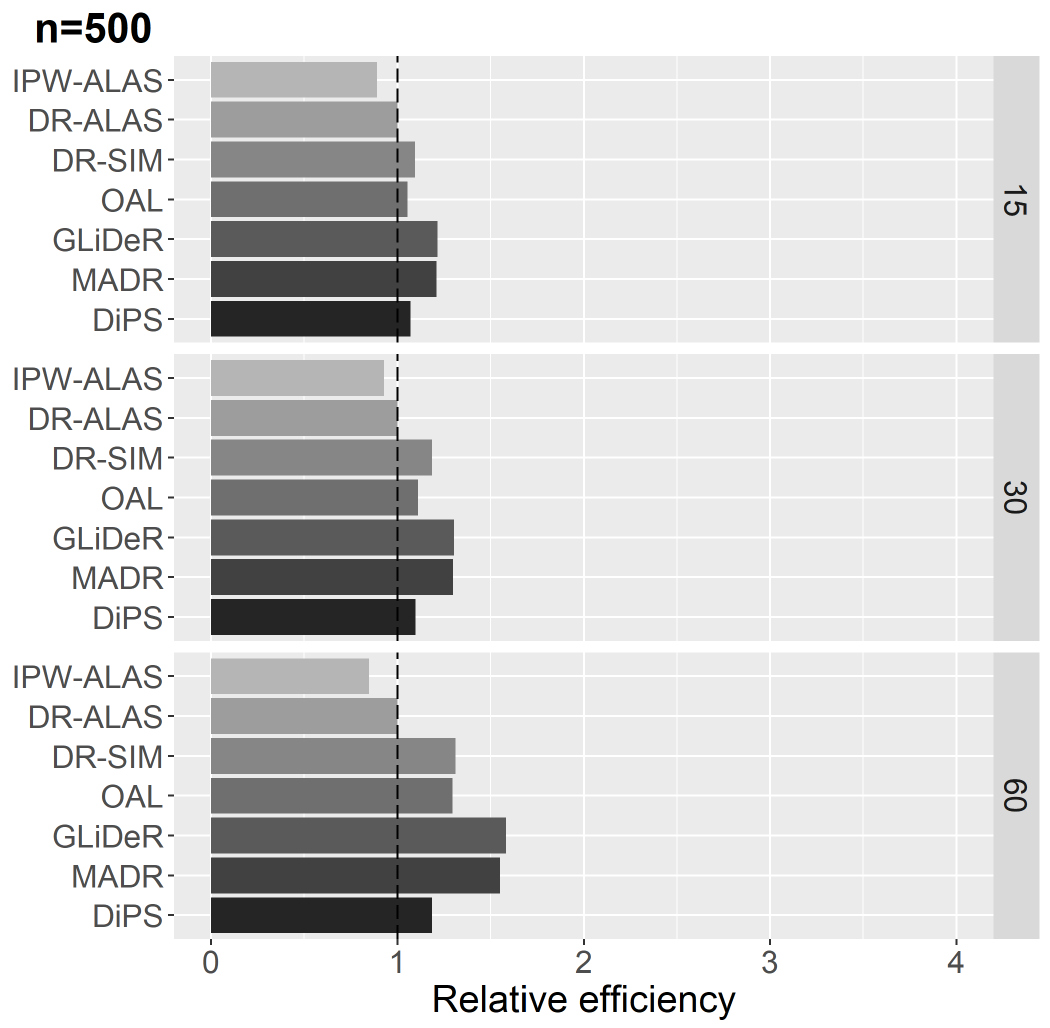} 
\includegraphics[scale=.425]{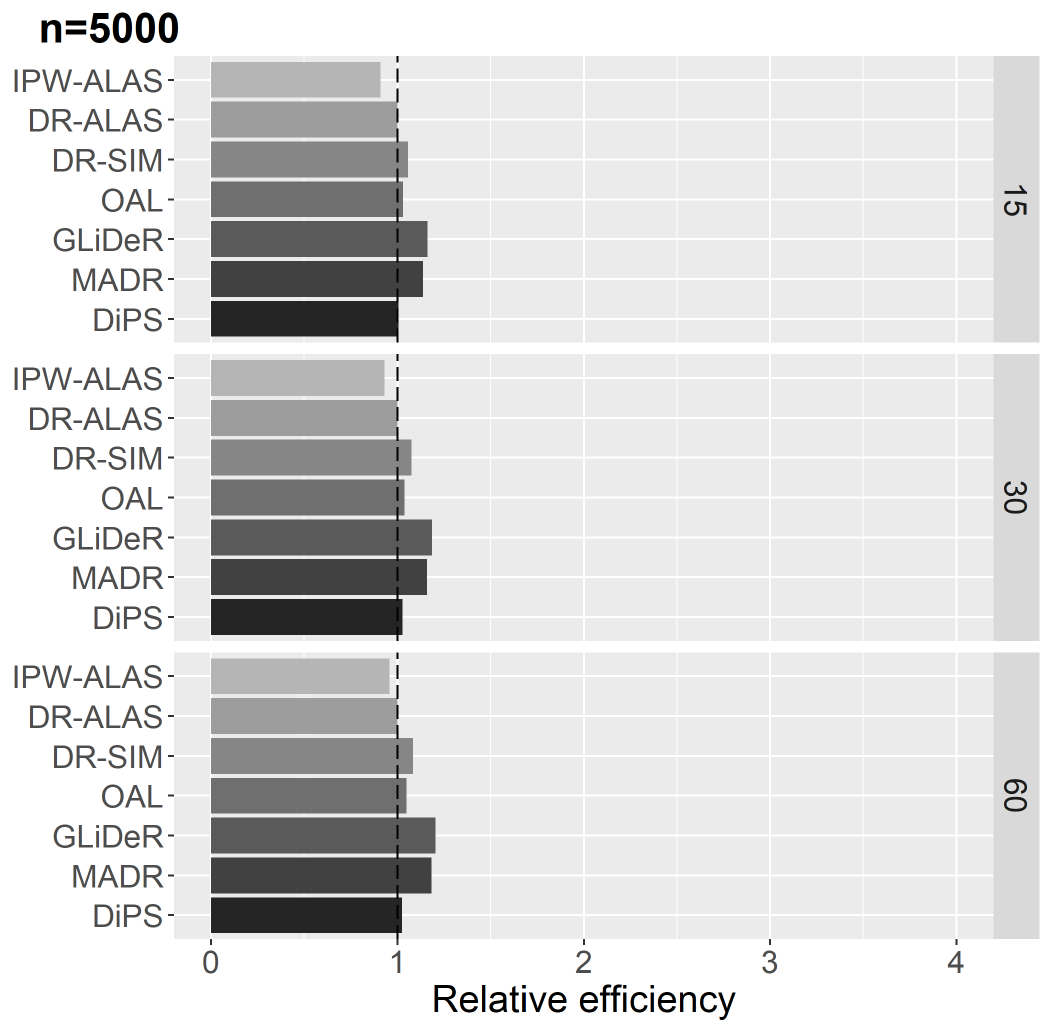}

\centerline{(b) Misspecified $\mu_k(\bx)$}
\includegraphics[scale=.425]{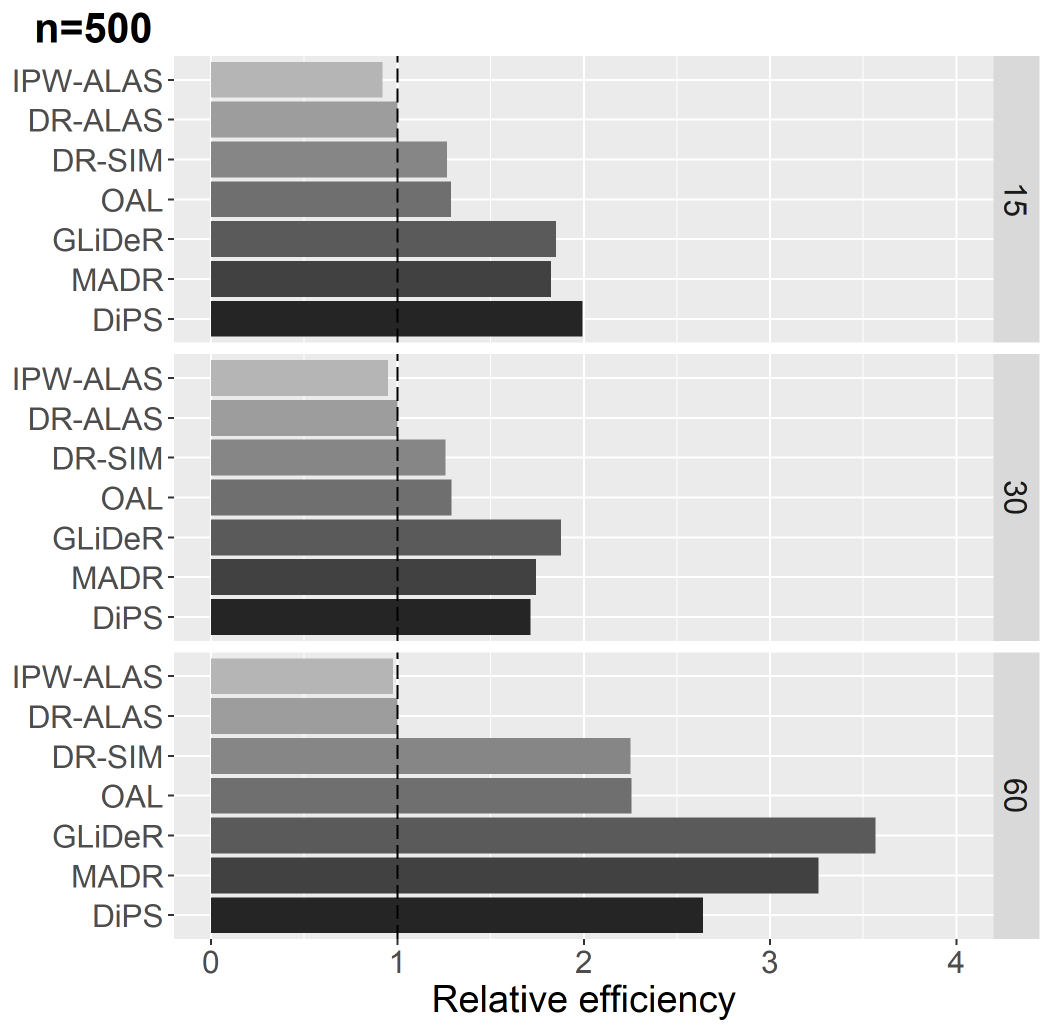}
\includegraphics[scale=.425]{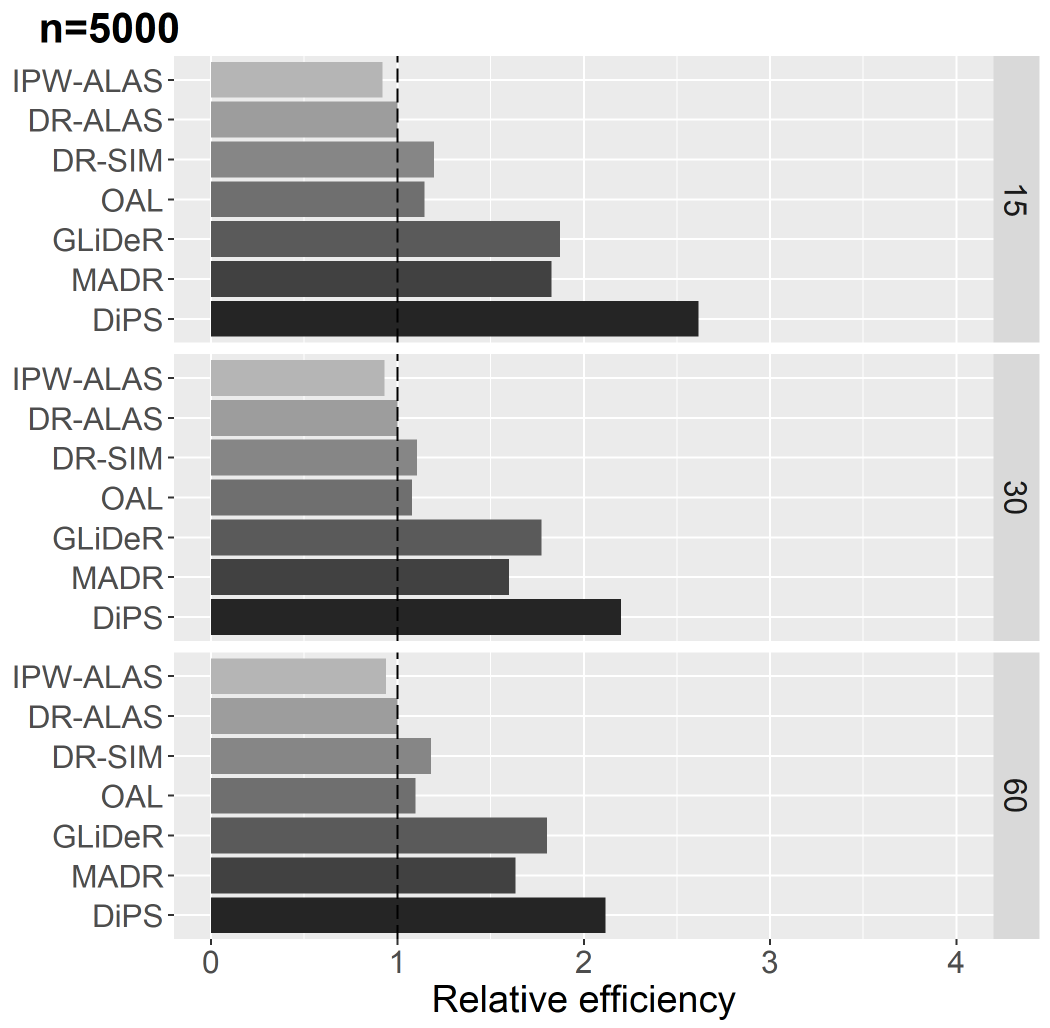}

\centerline{(c) Misspecified $\pi_1(\bx)$}
\includegraphics[scale=.425]{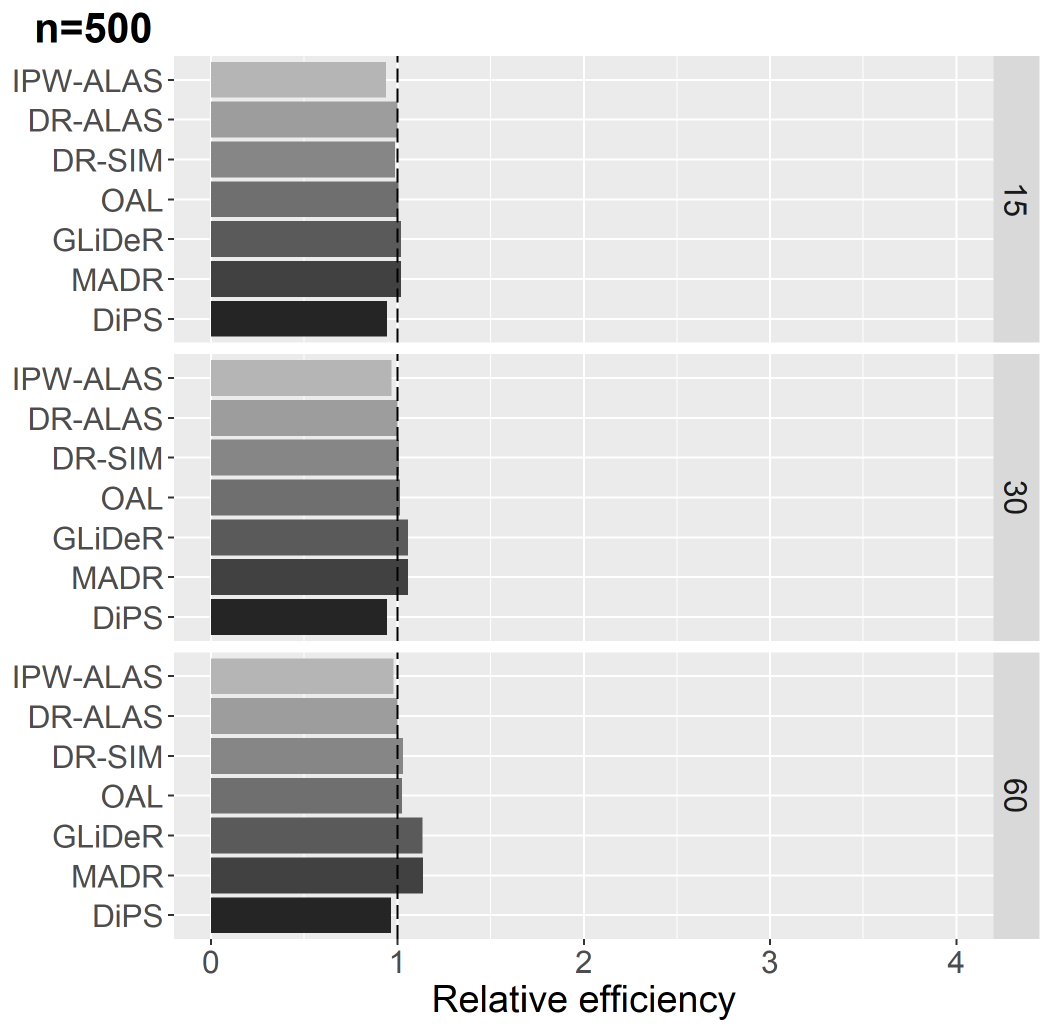}
\includegraphics[scale=.425]{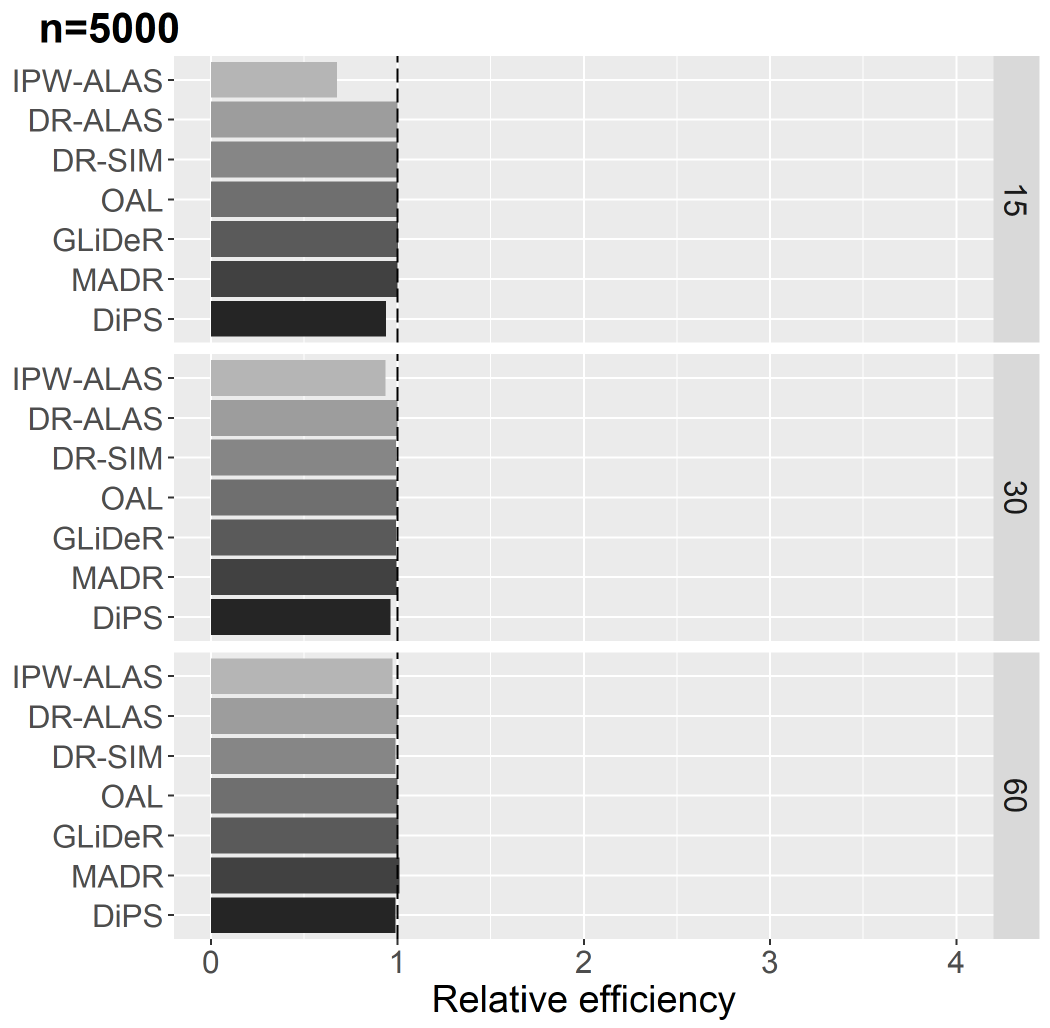}

\caption{\baselineskip=1pt RE relative to DR-ALAS by $n$, $p$, and specification scenario.
}
\label{fig:RE}%
\end{figure}

Table \ref{tab:cov} presents the performance of perturbation for DiPS when $p=15,30$ under correct working models.  SEs for DiPS were estimated using the MAD.  The empirical SEs (Emp SE), calculated from the sample standard deviations of $\Delthat$ over the simulation repetitions, were generally similar to the average of the SE estimates over the repetitions (ASE), despite some  overestimation up to 2-15\% of the Emp SE.
The coverage of the percentile CI's (Cover) were close to nominal 95\% levels but tended to be somewhat conservative.

\begin{table}[htbp]
\begin{tabular}{ccccc}
    \toprule
    \textbf{p} & \textbf{n} & \textbf{Emp SE} & \textbf{ASE} & \textbf{Cover} \\
    \midrule
    15    & 500   & 0.350 & 0.362 & 0.966 \\
    15    & 2500  & 0.151 & 0.167 & 0.970 \\
    15    & 5000  & 0.108 & 0.119 & 0.965 \\ \hline
    30    & 500   & 0.348 & 0.356 & 0.961 \\
    30    & 2500  & 0.150 & 0.167 & 0.975 \\
    30    & 5000  & 0.103 & 0.119 & 0.973 \\
    \bottomrule
\end{tabular}
\vspace{1.5em}
\caption{Perturbation performance under correctly specified models. Emp SE: empirical standard error over simulations, ASE: average of standard error estimates based on MAD over perturbations, Cover: Coverage of 95\% percentile intervals.}
\label{tab:cov}
\end{table}

\subsection{Data Example: Effect of Statins on Colorectal Cancer Risk in EMRs} \label{s:statin}
We applied DiPS to assess the effect of statins, a medication for lowering cholesterol levels, on the risk of colorectal cancer (CRC) among patients with inflammatory bowel disease (IBD) identified using data from EMRs of Partners Healthcare.  Previous studies have suggested that statins have a protective effect on CRC, but few studies have considered the effect specifically among IBD patients.  The EMR cohort consisted of $n=10,817$ IBD patients, including 1,375 statin users.  CRC status and statin use were ascertained by the presence of ICD9 diagnosis and prescription codes.  We adjusted for $p=15$ covariates as potential confounders, including age, gender, race, smoking status, indication of elevated inflammatory markers, examination with colonoscopy, use of biologics and immunomodulators, subtypes of IBD, disease duration, and presence of primary sclerosing cholangitis (PSC).  

For the working model $\Mscrmu$, we specified $g_{\mu}(u)=1/(1+e^{-u})$ to accomodate the binary outcome.  SEs for other estimators were obtained from the MAD over bootstrap resamples.
CIs were calculated from percentile intervals. We also calculated a two-sided p-value from a Wald test for the null that statins have no effect, using the point and SE estimates for each estimator.  The unadjusted estimate (None) based on difference in means by statins use was also calculated as a reference.  The left side of Table \ref{tab:data} shows that, without adjustment, the naive risk difference is estimated to be -0.8\% with a SE of 0.4\%.  The other methods estimated that statins had a protective effect ranging from around -1\% to -3\% after adjustment for covariates.  DiPS and DR-SIM were the most efficient estimators, with DiPS achieving estimated variance that ranged 34\% to 61\% lower than that of other estimators.

\begin{table}[htbp]
  \centering
  \scalebox{0.85}{
    \begin{tabular}{ccccccccc}
    \toprule
    \multicolumn{5}{c}{IBD EMR Study}     & \multicolumn{4}{c}{FOS} \\
          & \textbf{Est} & \textbf{SE} & \textbf{95\% CI} & \textbf{p-val} & \textbf{Est} & \textbf{SE} & \textbf{95\% CI} & \textbf{p-val} \\ \hline
    None  & -0.008 & 0.004 & (-0.017, 0) & 0.047 & 0.180 & 0.058 & (0.065, 0.298) & 0.002 \\
    IPW-ALAS & -0.022 & 0.004 & (-0.031, -0.015) & $<$0.001 & 0.182 & 0.063 & (0.053, 0.307) & 0.004 \\
    DR-ALAS & -0.020 & 0.005 & (-0.029, -0.012) & $<$0.001 & 0.140 & 0.063 & (0.031, 0.277) & 0.026 \\
    DR-SIM & -0.023 & 0.003 & (-0.029, -0.018) & $<$0.001 & 0.143 & 0.057 & (0.044, 0.257) & 0.013 \\
    OAL   & -0.008 & 0.004 & (-0.017, 0) & 0.048 & 0.175 & 0.061 & (0.062, 0.301) & 0.004 \\
    GLiDeR & -0.031 & 0.005 & (-0.04, -0.022) & $<$0.001 & 0.147 & 0.058 & (0.045, 0.258) & 0.012 \\
    MADR  & -0.030 & 0.005 & (-0.04, -0.021) & $<$0.001 & 0.149 & 0.056 & (0.037, 0.258) & 0.008 \\
    DiPS  & -0.024 & 0.003 & (-0.029, -0.017) & $<$0.001 & 0.141 & 0.058 & (0.039, 0.276) & 0.015 \\
    \bottomrule
    \end{tabular}
    }
  \vspace{1.5em}
\caption{Data example on the effect of statins on CRC risk in EMR data and the effect of smoking on logCRP in FOS data. Est: Point estimate, SE: estimated SE, 95\% CI: confidence interval, p-val: p-value from Wald test of no effect.
}
\label{tab:data}
\end{table}

\subsection{Data Example: Framingham Offspring Study}
The Framingham Offspring Study (FOS) is a cohort study initiated in 1971 that enrolled 5,124 adult children and spouses of the original Framingham Heart Study.  The study collected data over time on participants' medical history, physician examination, and laboratory tests to examine epidemiological and genetic risk factors of cardiovascular disease (CVD).  A subset of the FOS participants also have their genotype from the Affymetrix 500K single-nucleotide polymorphism (SNP) array available through the Framingham SNP Health Association Resource (SHARe) on dbGaP. We assessed the effect of smoking on C-reactive protein (CRP) levels, an inflammation marker highly predictive of CVD risk, while adjusting for potential confounders including gender, age, diabetes status, use of hypertensive medication, systolic and diastolic blood pressure measurements, and HDL and total cholesterol measurements, as well as a large number of SNPs in gene regions previously reported to be associated with inflammation or obesity.  While the inflmmation-related SNPs are not likely to impact smoking, we include them as prognostic covariates for efficiency.  The analysis includes $n=1,892$ individuals with available information on the CRP and the $p=121$ covariates, of which 113 were SNPs.

Since CRP is heavily skewed, we applied a log transformation so that the linear regression model in $\Mscrmu$ better fits the data.  SEs, CIs, and p-values were calculated in the same way as above.  The right side of Table \ref{tab:data} shows that different methods agree that smoking significantly increases logCRP.  In general, point estimates tended to attenuate after adjusting for covariates since smokers are likely to have other characteristics that increase inflammation.  DiPS, DR-SIM, and MADR were among the most efficient, though efficiency gains are tempered in this setting with larger $p$ relative to $n$.

\section{Discussion} \label{s:discussion}
In this paper we developed a novel IPW estimator for the ATE that accommodates data-driven variable selection through regularized regression.  The estimator retains double-robustness and is locally semiparametric efficient when $\nu=0$.  By calibrating the initial PS through smoothing, additional gains in efficiency can potentially be achieved in large samples under misspecification of the working outcome model.

In numerical studies, we used the extended regularized information criterion \citep{hui2015tuning} to tune adaptive LASSO, which maintains selection consistency when $log(p)/log(n) \to\nu$, for $\nu\in [0,1)$.  Other criteria such as cross-validation can also be used and may exhibit better performance in some cases.  To obtain a suitable bandwidth $h$, the bandwidth must be selected such that the dominating errors in the influence function, which are of order $O_p(n^{1/2}h^q + n^{-1/2}h^{-2})$, converges to $0$.  This is satisfied for $h= O(n^{-\alpha})$ for $\alpha \in (\frac{1}{2q},\frac{1}{4})$.  The optimal bandwidth $h^*$ is one that balances these bias and variance terms and is of order $h^*=O(n^{-1/(q+2)})$.  In practice we use a plug-in estimator $\hhat^* = \sighat n^{-1/(q+2)}$, where $\sighat$ is the sample standard deviation of either $\balphhat\trans\bX_i$ or $\bbetahat_k\trans\bX_i$, possibly after applying a monotonic transformation.  Cross-validation can also be used to select the the smoothing bandwidth.

The adaptive LASSO estimators $\balphhat$ and $\bbetahat_k$ are not uniformly root-$n$ consistent when the penalty is tuned to achieve consistent model selection \citep{potscher2009distribution}, and its oracle properties derived under fixed parameter asymptotics may fail to capture essential features of finite-sample distributions.  For example, they are not root-$n$ consistent when the true parameters are of order $O(n^{-1/2})$, if the true signals are relatively weak. 
The importance of uniform inference also been recently highlighted for treatment effect estimation in high-dimensional settings \citep{belloni2013inference,farrell2015robust}. It would be of interest to consider alternative variable selection approaches beyond those grounded in oracle properties to achieve uniform inference.
Another limitation of relying on adaptive LASSO is that when $p$ is large so that $\nu$ is large, a large power parameter $\gamma$ would be required to maintain the oracle properties, leading to an unstable penalty and poor finite sample performance.  
It would be of interest to consider modifications of the proposed procedure to accommodate high-dimensional settings with $p \gg n $ and more general sparsity assumptions in future work.

\backmatter

\section*{Acknowledgements}

The authors would like to thank the editor, associate editor, and two referees for their insightful feedback and suggestions.
Most of this work was done when the first author was a graduate student at Harvard University.  This work was supported by National Institutes of Health grants T32CA009337 and R01HL089778.
The views expressed in this article are those of the authors and do not necessarily reflect the views of the Department of Veterans Affairs.\vspace*{-8pt}

\bibliographystyle{biom} 
\bibliography{mybibilo}

\section*{Supporting Information}

Web Appendices referenced in Sections \ref{s:method}, \ref{s:asymptotics}, and \ref{s:numerical} as well as the R code 
implementing the procedure are available with this paper at the Biometrics website on Wiley Online Library.

\processdelayedfloats

\newpage

\appendix

\section*{Supporting Information for ``Estimating Average Treatment Effects with a Double-Index Propensity Score'' by David Cheng, Abhishek Chkrabortty, Ashwin N. Ananthakrishnan, and Tianxi Cai.}

\newpage

These supplementary materials describe the requisite regularity conditions (Web Appendix A) and provides derivations of the two theorems in the main text (Web Appendix B).  Web Appendix C reports the correlation matrix used for the covariates in the simulations.   
Web Appendix D reports simulation results on the proportion of observations with negative values of DiPS across different scenarios.

The following notations will facilitate the derivations.  Throughout this Web Appendix, we suppress the $k$ in $\bbeta_k$, $\bbetahat_k$, $\bbetabar_k$, $\bthet_k$, $\bthethat_k$, and $\bthetbar_k$ for ease of notation but implicitly understand these quantities to be defined with respect to treatment $k=0,1$ in general. Let $\bSbar = (\balphbar\trans\bX,\bbetabar\trans\bX)\trans$ be $\bX$ in the directions of $\balphbar$ and $\bbetabar$, regardless of the adequacy of the working models.  Let the true density of $\bSbar$ at $\bs$ be $f(\bs)$, the propensity score given $\bSbar=\bs$ for $k=0,1$ be $\pi_k(\bs)=P(T=k\mid \bSbar=\bs)$, and $l_k(\bs) = \pi_k(\bs)f(\bs)$.  Given a $\bx \in \mathbb{R}^{p}$, $\balph,\bbeta\in \mathbb{R}^p$, for $\bthet = (\balph\trans,\bbeta\trans)\trans$, let:
\begin{align}
\pihat_k(\bx;\bthet) = \pihat_k(\bx;\balph,\bbeta) = \frac{\lhat_k(\bx;\bthet)}{\fhat(\bx;\bthet)} = \frac{\sum_{j=1}K_h\{(\balph,\bbeta)\trans(\bX_j-\bx)\}I(T_i=k)}{\sum_{j=1}K_h\{(\balph,\bbeta)\trans(\bX_j-\bx)\}}.
\end{align}

For a $p$ length random vector $\bV$, let $\bV^{\dagger} = (\bV,\bzero_{p})$ be the $p\times 2$ matrix of the vector augmented by column of zeros on the right and $\bV^{\ddagger} = (\bzero_{p},\bV)$ similarly by a column of zeros on the left.  For any two vectors $\bV_i$ and $\bV_j$, let $\bV_{ji} = \bV_j - \bV_i$.  Let $K(\bu)$ be a bivariate symmetric kernel function of order $q>2$, with a finite $q$-th moment.  Let $\Kdot(\bu) = \partial K(\bu)/\partial\bu$ and $\Kdot_h(\bv) = h^{-3}\Kdot(\bv/h)$.  For any vector $\bV$ of length $p$ and $\Ascr \subseteq \{ 1,2,\ldots,p\}$, with $\abs{\Ascr} = p_0$, let $\bV_{\Ascr}$ denote a $p_0$-length vector that is $\bV$ restricted to coordinates indexed in $\Ascr$.  Similarly, let $\bV\trans_{\Ascr}$ denote $\bV\trans$ restricted to coordinates indexed in $\Ascr$.

\section*{Web Appendix A: Regularity Conditions}
\label{wa:assump}
(i) $K(\bu)$ is a bivariate kernel function of order $q > 2$, with a finite $q$-th moment. (ii) $K(\bu)$ is bounded and continuously differentiable with a compact support. (iii) $\Kdot(\bu)$ is bounded, integrable, and Lipshitz continuous. (iv) $\Xscr$ is compact. (v) $f(\bs)$ is bounded and bounded away from $0$ over its support. (vi) $f(\bs)$, $\pi_k(\bs)$, and $E(Y|\bSbar=\bs,T=k)$ for $k=0,1$ are $q$-times continuously differentiable. (vii) $E(\bX | \bSbar =\bs)$, $E(\bX|\bSbar=\bs,T=k)$, and $E(\bX Y | \bSbar = \bs, T=k)$ are continuously differentiable for $k=0,1$. (viii) There exists $0<k_1 < k_2 <\infty$ such that the minimum and maximum eigenvalues of $\frac{1}{n}\sum_{i=1}^n\bX_i\bX_i\trans$ around bounded below by $k_1$ and above by $k_{2}$. (ix) $\Thetpi$ and $\Thetmu$ are compact. (x) For all $u\in\mathbb{R}$, $1/M \leq g_{\mu}'(u) \leq M$ and $\abs{g_{\mu}''(u)} \leq M$ and for some $0<M<\infty$.

\section*{Web Appendix B: Derivations of Theorems 1 and 2} \label{wa:proof}

\subsection*{Supporting Lemmas} \label{wa:supplem}
Lemma \ref{lm:mnord} identifies the stochastic order of a standardized mean when the variance of the observations is of a known order.  It will be useful for controlling certain terms that will emerge in the expansion.  Lemma \ref{lm:unicv} shows the uniform convergence rate for kernel smoothing when $\balph$ and $\bbeta$ are fixed, which is a fundamental result used in our approach. Lemma \ref{lm:avgrd} simplifies the average of the gradients of the average of terms that are inversely weighted by the calibrated PS evaluated at the least false parameters.  These terms appear repeatedly in subsequent derivations.
\begin{lemma}\label{lm:mnord}
Let $\{ X_{i,n}\}$ be a triangular array such that $X_{1,n},\ldots,X_{n,n}$ are iid for each $n\in \mathbb{N}$.  Suppose that $\sigma^2_n = Var(X_{i,n}) = O(c_n^2)$, where $c_n$ is some positive sequence.  Then:
\begin{align}
n^{1/2}\abs{\bar{X}_n - \mu_n} \leq O_p(c_n),
\end{align}
where $\bar{X}_n = n^{-1}\sum_{i=1}^n X_{i,n}$ and $\mu_n = E(X_{i,n})$.
\end{lemma}
\begin{proof}
By Chebyshev's inequality, for any $k > 0$:
\begin{align}
\P\left( n^{1/2}\abs{\bar{X}_n -\mu_n}/c_n \geq k \right) \leq \frac{\sigma_n^2}{c_n^2 k^2}.
\end{align}
Let $M = \sup_{n\in\mathbb{N}} \sigma^2_n/c_n^2$.  For any $\epsilon>0$, the desired result is obtained by taking $k = (M/\epsilon)^{1/2}$.
\end{proof}

\begin{lemma}\label{lm:unicv}
The uniform convergence rate for two-dimensional smoothing over $\bX$ in the directions $\balphbar$ and $\bbetabar$ is given by:
\begin{align}
\sup_{\bx} \norm{\pihat_k(\bx;\bthetbar) - \pi_k(\bx;\bthetbar)} = O_p(a_n),
\end{align}
where $\bthetbar = (\balphbar\trans,\bbetabar\trans)\trans$, $\pi_k(\bx;\bthetbar) = P(T=k\mid \balphbar\trans\bX=\balphbar\trans\bx,\bbetabar\trans\bX=\bbetabar\trans\bx)$, and:
\begin{align}
a_n = h^q + \{ log(n) / (nh^2)\}^{1/2}.
\end{align}
\begin{proof}
Smoothing over $\bX$ in the directions of $\balphbar$ and $\bbetabar$ is the same as a two-dimensional kernel smoothing since $\balphbar$ and $\bbetabar$ are fixed.  See, for example, \cite{hansen2008uniform} for the derivation of uniform convergence rates for $d$-dimensional smoothing.
\end{proof}
\end{lemma}

\begin{lemma} \label{lm:avgrd}
Let $g(\bZ)$ denote a real-valued square-integrable transformation of the data $\bZ = (\bX,T,Y)\trans$.  Under the above regularity conditions and that $E\{ g(\bZ)|\bSbar = \bs\}$ and $E\{ \bX g(\bZ)|\bSbar = \bs\}$ are continuous in $\bs$:
\begin{align}
&n^{-1}\sum_{i=1}^n \ddbalphT \frac{g(\bZ_i)}{\pihat_k(\bX_i;\bthetbar)} = \E\left[ \Kdot_h(\bSbar_{ji})\trans \left\{ \pi_k(\bSbar_i) - I(T_j=k)\right\}\frac{g(\bZ_i)}{\pi_k(\bSbar_i)l_k(\bSbar_i)}\bXdagT_{ji}\right] + O_p(b_n)\\
&n^{-1}\sum_{i=1}^n \ddbbetaT \frac{g(\bZ_i)}{\pihat_k(\bX_i;\bthetbar)} = \E\left[ \Kdot_h(\bSbar_{ji})\trans \left\{ \pi_k(\bSbar_i) - I(T_j=k)\right\}\frac{g(\bZ_i)}{\pi_k(\bSbar_i)l_k(\bSbar_i)}\bXddagT_{ji}\right] + O_p(b_n),
\end{align}
where $b_n = n^{-1/2}h^{-1} + n^{-1}h^{-3}$, for $k=0,1$.
\end{lemma}
\begin{proof}
We will show the first equality for the gradient with respect to $\balph$, with the second equality being analogous.  First note each of the gradients can be written:
\begin{align}
\ddbalphT \frac{1}{\pihat_k(\bX_i;\bthetbar)} &= \frac{\ddbalphT \fhat(\bX_i;\bthetbar)\lhat_k(\bX_i;\bthetbar) - \fhat(\bX_i;\bthetbar) \ddbalphT\lhat_k(\bX_i;\bthetbar)}{\lhat_k(\bX_i;\bthetbar)^2} \\
&= n^{-1}\sum_{j=1}^n \Kdot_h(\bSbar_{ji})\trans \frac{\lhat_k(\bX_i;\bthetbar) - I(T_j = k)\fhat(\bX_i;\bthetbar)}{\lhat_k(\bX_i;\bthetbar)^2}\bXdagT_{ji}.
\end{align}
Consequently, the average of the gradients can be written:
\begin{align}
&n^{-1}\sum_{i=1}^n \ddbalphT \frac{g(\bZ_i)}{\pihat_k(\bX_i;\bthetbar)} = n^{-2}\sum_{i,j} \Kdot_h(\bSbar_{ji})\trans \frac{\lhat_k(\bX_i;\bthetbar) - I(T_j = k)\fhat(\bX_i;\bthetbar)}{\lhat_k(\bX_i;\bthetbar)^2}\bXdagT_{ji} g(\bZ_i) \\
&\qquad= n^{-2}\sum_{i,j} \Kdot_h(\bSbar_{ji})\trans \frac{l_k(\bSbar_i) - I(T_j = k)f(\bSbar_i)}{\lhat_k(\bX_i;\bthetbar)^2}\bXdagT_{ji} g(\bZ_i) + O_p(a_n \Escr_n) \\
&\qquad= n^{-2}\sum_{i,j} \Kdot_h(\bSbar_{ji})\trans \frac{l_k(\bSbar_i) - I(T_j = k)f(\bSbar_i)}{\l_k(\bSbar_i)^2}\bXdagT_{ji} g(\bZ_i) + O_p(a_n \Escr_n),
\end{align}
where we make repeated use of the uniform convergence of $\lhat_k(\bX_i;\bthetbar)$ and $\fhat(\bX_i;\bthetbar)$ to $l_k(\bSbar_i)$ and $f(\bSbar_i)$, $\Escr_n$ is a term of the same order as the main term so that $O_p(a_n\Escr_n)$ will be a negligible lower-order term, and use that $l_k(\bs)$ is bounded over its support in the last equality.  To facilitate application of the V-statistic projection lemma, define:
\begin{align}
&\bm_{1,k}(\bZ_j) = \E_{\bZ_i}\left\{ \Kdot_h(\bSbar_{ji})\trans \frac{l_k(\bSbar_i)-I(T_j=k)f(\bSbar_i)}{l_k(\bSbar_i)^2}\bXdagT_{ji}g(\bZ_i)\right\} \\
&\bm_{2,k}(\bZ_i) = \E_{\bZ_j}\left\{ \Kdot_h(\bSbar_{ji})\trans \frac{l_k(\bSbar_i)-I(T_j=k)f(\bSbar_i)}{l_k(\bSbar_i)^2}\bXdagT_{ji}g(\bZ_i)\right\}\\
&\bm_{k} = \E\left\{ \Kdot_h(\bSbar_{ji})\trans \frac{l_k(\bSbar_i)-I(T_j=k)f(\bSbar_i)}{l_k(\bSbar_i)^2}\bXdagT_{ji}g(\bZ_i)\right\}\\
&\bepsi_{1,k} = n^{-1}\E\norm{\Kdot_h(\bSbar_{ii})\trans \frac{l_k(\bSbar_i)-I(T_i=k)f(\bSbar_i)}{l_k(\bSbar_i)^2}\bXdagT_{ii}g(\bZ_i)} = 0 \\
&\bepsi_{2,k} = n^{-1}\left( \E\left[ \left\{\Kdot_h(\bSbar_{ji})\trans\frac{l_k(\bSbar_i)-I(T_j=k)f(\bSbar_i)}{l_k(\bSbar_i)^2}\bXdagT_{ji}g(\bZ_i)\right\}^2\right]\right)^{1/2}
\end{align}
We now further evaluate each term.  The first term can be simplified through a change-of-variables:
\begin{align}
&\bm_{1,k}(\bZ_j) = \E_{\bSbar_i}\left[ \Kdot_h(\bSbar_{ji})\trans\left\{1 - \frac{I(T_j=k)}{\pi_k(\bSbar_i)} \right\}\frac{1}{l_k(\bSbar_i)}\E\left\{ \bXdagT_{ji}g(\bZ_i)\mid\bSbar_i\right\}\right] \\
&= \int \Kdot_h(\bSbar_j - \bs_1)\trans \left\{1 - \frac{I(T_j=k)}{\pi_k(\bs_1)} \right\}\frac{1}{\pi_k(\bs_1)}\E\left\{ \bXdagT_{ji}g(\bZ_i)\mid\bSbar_i=\bs_1\right\}d\bs_1 \\
&= h^{-1}\int \Kdot(\bpsi_j)\trans \left\{1 - \frac{I(T_j=k)}{\pi_k(h\bpsi_j + \bSbar_j)} \right\}\frac{1}{\pi_k(h\bpsi_j + \bSbar_j)}\E\left\{ \bXdagT_{ji}g(\bZ_i)\mid\bSbar_i=h\bpsi_j + \bSbar_j\right\}d\bpsi_j \\
&= O_p(h^{-1}),
\end{align}
where the last step follows from bounding the integrand.  Similarly for the second term:
\begin{align}
&\bm_{2,k}(\bZ_i) = \E_{\bSbar_j}\left[\Kdot_h(\bSbar_{ji})\trans \E\left\{ (1 - \frac{I(T_j=k)}{\pi_k(\bSbar_i)})\bXddagT_{ji} \mid \bSbar_j \right\}\frac{g(\bZ_i)}{l_k(\bSbar_i)}\right] \\
&= \int \Kdot_h(\bs_2 - \bSbar_i)\trans \E\left[ \left\{1 - \frac{I(T_j=k)}{\pi_k(\bSbar_i)}\right\}\bXddagT_{ji} \mid \bSbar_j = \bs_2 \right]\frac{g(\bZ_i)}{l_k(\bSbar_i)}f(\bs_2)d\bs_2 \\
&= h^{-1}\int \Kdot(\bpsi_i)\trans \E\left[ \left\{1 - \frac{I(T_j=k)}{\pi_k(\bSbar_i)}\right\}\bXddagT_{ji} \mid \bSbar_j = h\bpsi_i+\bSbar_i \right]\frac{g(\bZ_i)}{l_k(\bSbar_i)}f(h\bpsi_i+\bSbar_i)d\bpsi_i \\
&= O_p(h^{-1}),
\end{align}
where again the last step follows from bounding the integrand.  Now,
$\bepsi_{2,k} = O_p(n^{-1}h^{-3})$ from bounding the terms in the expectation, except for $g(\bZ_i)$.  The projection lemma thus yields:
\begin{align}
&n^{-2}\sum_{i,j} \Kdot_h(\bSbar_{ji}) \frac{l_k(\bSbar_i) - I(T_j = k)f(\bSbar_i)}{\l_k(\bSbar_i)^2}\bXdagT_{ji} g(\bZ_i)  \\
&\qquad = \bm_k + n^{-1}\sum_{j=1}^n \bm_{1,k}(\bZ_j) - \bm_k + n^{-1}\sum_{i=1}^n \bm_{2,k}(\bZ_i)-\bm_k + O_p(\bepsi_1 + \bepsi_2) \\
&\qquad = \bm_k + O_p(n^{-1/2}h^{-1}) + O_p(h^{n^{-1}h^{-3}}),
\end{align}
for $k=0,1$, where the last line follows from application of Lemma \ref{lm:mnord}.  Re-arrangement of terms and collecting the dominant errors yield the desired result.
\end{proof}

\subsection*{Expansion of Normalization Constant} \label{wa:nrmconst}
We will first show the normalization constant is $1$ up to some lower order terms, which will allow us to account for the normalization in the expansion.  The approach for the analysis parallels that of the main expansion.  First note that:
\begin{eq}
n^{-1}\sum_{i=1}^n \ipwhati &= n^{-1}\sum_{i=1}^n \ipwbari + n^{-1}\sum_{i=1}^n \left\{ \frac{1}{\pihat_k(\bX_i;\bthetbar)} - \frac{1}{\pi_k(\bX_i;\bthetbar)}\right\}I(T_i = k) \\
&\qquad + n^{-1}\sum_{i=1}^n \left\{ \frac{1}{\pihat_k(\bX_i;\bthethat)} - \frac{1}{\pihat_k(\bX_i;\bthetbar)}\right\} I(T_i = k) \\
&= \Vscrhat_{1,k} + \Vscrhat_{2,k} + \Vscrhat_{3,k},
\end{eq}
where:
\begin{eq}
&\Vscrhat_{1,k} = n^{-1}\sum_{i=1}^n \ipwbari, \quad \Vscrhat_{2,k} = n^{-1}\sum_{i=1}^n \left\{ \frac{1}{\pihat_k(\bX_i;\bthetbar)} - \frac{1}{\pi_k(\bX_i;\bthetbar)}\right\}I(T_i = k), \\
&\Vscrhat_{3,k} = n^{-1}\sum_{i=1}^n \left\{ \frac{1}{\pihat_k(\bX_i;\bthethat)} - \frac{1}{\pihat_k(\bX_i;\bthetbar)}\right\} I(T_i = k).
\end{eq}

The second term is of order:
\begin{eq}
\abs{\Vscrhat_{2,k}} &= n^{-1}\abs{\sum_{i=1}^n \frac{\pi_k(\bX_i;\bthetbar)-\pihat_k(\bX_i;\bthetbar)}{\pi_k(\bX_i;\bthetbar)\pihat_k(\bX_i;\bthetbar)}I(T_i=k)} \\
&\leq \sup_{\bX_i}\abs{\pihat_k(\bX_i;\bthetbar)-\pi_k(\bX_i;\bthetbar)} n^{-1}\sum_{i=1}^n \abs{\frac{I(T_i = k)}{\pihat_k(\bX_i;\bthetbar)\pi_k(\bX_i;\bthetbar)}} \\
&= O_p(a_n),
\end{eq}
where the last step follows from uniform convergence of $\pihat_k(\bX_i;\bbetabar)$ to $\pi_k(\bX_i;\bthetbar)$ and noting the remaining sum is $O_p(1)$ plus some lower-order term.  The third term can be written:
\begin{eq}
\Vscrhat_{3,k} &= n^{-1}\sum_{i=1}^n \left\{ \frac{1}{\pihat_k(\bX_i;\balphhat,\bbetahat)} - \frac{1}{\pihat_k(\bX_i;\balphbar,\bbetahat)} + \frac{1}{\pihat_k(\bX_i;\balphbar,\bbetahat)} - \frac{1}{\pihat_k(\bX_i;\balphbar,\bbetabar)}\right\}I(T_i=k) \\
&= n^{-1}\sum_{i=1}^n \left\{\ddbalphT \frac{\balphhat-\balphbar}{\pihat_k(\bX_i;\balphbar;\bbetahat)} + \ddbbetaT \frac{\bbetahat-\bbetabar}{\pihat_k(\bX_i;\balphbar,\bbetabar)}\right\}I(T_i =k) + O_p(n^{-1}\Escr_{\balph,n} + n^{-1}\Escr_{\bbeta,n}) \\
&= n^{-1}\sum_{i=1}^n \left\{\ddbalphT \frac{\balphhat-\balphbar}{\pihat_k(\bX_i;\balphbar;\bbetabar)} + \ddbbetaT \frac{\bbetahat-\bbetabar}{\pihat_k(\bX_i;\balphbar,\bbetabar)}\right\}I(T_i = k) \\
&\qquad\qquad\qquad\qquad + O_p(n^{-1}\Escr_{\balph,n} + n^{-1}\Escr_{\bbeta,n} + n^{-1}\Escr_{\balph\bbeta,n}),
\end{eq}
where the last equality uses that that $\Kdot(\bu)$ is Lipshitz continuous and that $\Escr_{\balph,n}$, $\Escr_{\bbeta,n}$, and $\Escr_{\balph\bbeta,n}$ are terms of the same order as $n^{-1}\sum_{i=1}^n \ddbalphT \pihat_k(\bX_i;\balphbar;\bbetabar)^{-1}$ so that the error terms will be negligible lower-order terms.  Applying Lemma \ref{lm:avgrd}, we can simplify:
\begin{align}
n^{-1}\sum_{i=1}^n \ddbalphT \frac{I(T_i =k )}{\pihat_k(\bX_i;\balphbar;\bbetabar)} &= \E\left[ \Kdot_h(\bSbar_{ji})\trans \left\{ \pi_k(\bSbar_i) - I(T_j=k)\right\}\frac{I(T_i=k)}{\pi_k(\bSbar_i)l_k(\bSbar_i)}\bXdagT_{ji}\right] + O_p(b_n).
\end{align}
Further simplifying the expectation we have:
\begin{align}
&\E\left[ \Kdot_h(\bSbar_{ji})\trans \left\{ \pi_k(\bSbar_i) - I(T_j=k)\right\}\frac{I(T_i=k)}{\pi_k(\bSbar_i)l_k(\bSbar_i)}\bXdagT_{ji}\right] \\
&\quad= \E\left( \frac{\Kdot_h(\bSbar_{ji})}{l_k(\bSbar_i)}\trans \left[ \pi_k(\bSbar_i)\left\{ \E(\bXdagT_j\mid \bSbar_j) - \E(\bXdagT_i\mid \bSbar_i, T_i=k)\right\} \right.\right. \\
&\quad\quad\left.\left. - \pi_k(\bSbar_j)\left\{ \E(\bXdagT_j\mid\bSbar_j,T_j=k)- \E(\bXdagT_i\mid\bSbar_i,T_i=k)\right\} \vphantom{\pi_k(\bSbar_i)}\right]\vphantom{\frac{\Kdot_h(\bSbar_{ji})}{l_k(\bSbar_i)}} \right) \\
&\quad = h^{-1}\iint\Kdot(\bpsi_1) \frac{f(h\bpsi_1 + \bs_1)}{\pi_k(\bs_1)}\left[ \pi_k(\bs_1)\left\{ \E(\bX_j\mid\bSbar_j=h\bpsi_1 + \bs_1) - \E(\bXdagT_i\mid\bSbar_i=\bs_1,T_i=k)\right\} \right.\\
&\quad\quad -\left. \pi_k(h\bpsi_1 + \bs_1)\left\{ \E(\bXdagT_j\mid \bSbar_j = h\bpsi_1 + \bs_1, T_j = k)- \E(\bXdagT_i \mid \bSbar_i=\bs_1, T_i = k)\right\}\right]d\bpsi_1d\bs_1 \\
&= O(h^{-1}),
\end{align}
where the last step follows from bounding terms in the integrand.  Similarly:
\begin{align}
n^{-1}\sum_{i=1}^n \ddbbetaT \frac{I(T_i = k)}{\pihat_k(\bX_i;\balphbar,\bbetabar)} =  O_p(h^{-1}) + O_p(b_n).
\end{align}  

Collecting all the results:
\begin{align}
n^{-1}\sum_{i=1}^n \frac{I(T_i = k)}{\pihat_k (\bX_i;\bthethat)} &= 1+O_p(n^{-1/2}) + O_p(a_n) + O_p(n^{-1/2}h^{-1}) + O_p(n^{-1/2}b_n) \\
&= 1 + O_p(a_n).
\end{align}

\subsection*{Main Results} \label{wa:mainexp}
The approach for the expansion will be to decompose $\Wscrhat_k$ into terms representing the variability contributed from smoothing, with known $\bthet$, and from estimating $\bthet$.  The term contributed from smoothing is written in terms of a V-statistic and analyzed using a V-statistic projection lemma (Lemma 8.4 of \cite{newey1994large}).  The term contributed from estimating $\bthet$ is analyzed applying arguments for the oracle properties of adaptive LASSO estimators from \cite{zou2009adaptive},\cite{hui2015tuning}, and \cite{lu2012robustness}.  First note that:
\begin{eq}
\Wscrhat_k &= \left\{ n^{-1}\sum_{i=1}^n \ipwhati \right\}^{-1}\left\{ n^{-1/2}\sum_{i=1}^n \ipwhati (Y_i - \mubar_k)\right\} \\
&= n^{-1/2}\sum_{i=1}^n \ipwhati (Y_i - \mubar_k) + \left[\left\{ 1+O_p(a_n)\right\}^{-1} -1\right] n^{-1/2}\sum_{i=1}^n \ipwhati (Y_i - \mubar_k) \\
&= n^{-1/2}\sum_{i=1}^n \ipwhati (Y_i - \mubar_k) \left\{ 1+ O_p(a_n)\right\},
\end{eq}
where the second step follows from the result in Web Appendix C.  Define:
\begin{eq}
\Wscrtld_k &= n^{-1/2}\sum_{i=1}^n \ipwhati (Y_i - \mubar_k) = \Wscrtld_{1,k} + \Wscrtld_{2,k} + \Wscrtld_{3,k},
\end{eq}
where:
\begin{eq}
&\Wscrtld_{1,k} = n^{-1/2}\sum_{i=1}^n \ipwbari (Y_i - \mubar_k) \\
&\Wscrtld_{2,k} = n^{-1/2}\sum_{i=1}^n \left\{ \frac{1}{\pihat_k(\bX_i;\bthetbar)} - \frac{1}{\pi_k(\bX_i;\bthetbar)}\right\}I(T_i=k)(Y_i - \mubar_k) \\
&\Wscrtld_{3,k} = n^{-1/2}\sum_{i=1}^n \left\{ \frac{1}{\pihat_k(\bX_i;\bthethat)} - \frac{1}{\pihat_k(\bX_i;\bthetbar)}\right\}I(T_i=k)(Y_i - \mubar_k).
\end{eq}

We now proceed to further expand the second and third terms.  For the second term:
\begin{align}
&\Wscrtld_{2,k} 
= -n^{-1/2}\sum_{i=1}^n \frac{\lhat(\bX_i;\bthetbar)-\fhat(\bX_i;\bthetbar)\pi_k(\bSbar_i)}{l_k(\bX_i;\bthetbar)}\ipwbari(Y_i - \mubar_k) \\
&\quad + -n^{-1/2}\sum_{i=1}^n \left\{ \frac{1}{\lhat_k(\bX_i;\bthetbar)} - \frac{1}{l_k(\bX_i;\bthetbar)}\right\}\left\{ \lhat(\bX_i;\bthetbar)-\fhat(\bX_i;\bthetbar)\pi_k(\bSbar_i) \right\}\ipwbari(Y_i - \mubar_k) \\
&= -n^{-1/2}\sum_{i=1}^n \frac{\lhat(\bX_i;\bthetbar)-\fhat(\bX_i;\bthetbar)\pi_k(\bSbar_i)}{l_k(\bX_i;\bthetbar)}\ipwbari(Y_i - \mubar_k)  + O_p(n^{1/2}a_n^2),
\end{align}
where the last equality follows from repeated use of uniform convergence of of $\lhat(\bX_i;\bthetbar)$ and $\fhat(\bX_i;\bthetbar)$ to $l(\bX_i;\bthetbar)$ and $f(\bX_i;\bthetbar)$ and that $n^{-1/2}\sum_{i=1}^n \ipwbari\frac{Y_i - \mubar}{l_k(\bX_i;\bthetbar)} = O_p(n^{1/2})$. Thus:
\begin{align}
\Wscrtld_{2,k} &= -n^{1/2} n^{-2}\sum_{i,j} K_h(\bSbar_{ji}) \left\{ I(T_j=k)-\pi_k(\bSbar_i)\right\} \ipwbari \frac{Y_i - \mubar_k}{l_k(\bSbar_i)} + O_p(n^{1/2}a_n^2) \\
&= \Wscrtld_{ct,2,k} + \Wscrtld_{nc,2,k} + O_p(n^{1/2}a_n^2),
\end{align}
with a centered and a non-centered V-statistic:
\begin{align}
&\Wscrtld_{ct,2,k} = -n^{1/2} n^{-2}\sum_{i,j} K_h(\bSbar_{ji}) \left\{ I(T_j=k)-\pi_k(\bSbar_j)\right\} \ipwbari \frac{Y_i - \mubar_k}{l_k(\bSbar_i)} \\
&\Wscrtld_{nc,2,k} = -n^{1/2} n^{-2}\sum_{i,j} K_h(\bSbar_{ji}) \left\{ \pi_k(\bSbar_i)-\pi_k(\bSbar_i)\right\} \ipwbari \frac{Y_i - \mubar_k}{l_k(\bSbar_i)}.
\end{align}

To facilitate application of the projection lemma, let:
\begin{align}
&m_{1,ct,2,k}(\bZ_j) = \E_{\bZ_i}\left[ K_h(\bSbar_{ji}) \left\{ I(T_j=k)-\pi_k(\bSbar_j)\right\} \ipwbari \frac{Y_i - \mubar_k}{l_k(\bSbar_i)} \right]\\
&m_{2,ct,2,k}(\bZ_i) = \E_{\bZ_j}\left[ K_h(\bSbar_{ji}) \left\{ I(T_j=k)-\pi_k(\bSbar_j)\right\} \ipwbari \frac{Y_i - \mubar_k}{l_k(\bSbar_i)} \right] \\
&m_{ct,2,k} = \E\left[ K_h(\bSbar_{ji}) \left\{ I(T_j=k)-\pi_k(\bSbar_j)\right\} \ipwbari \frac{Y_i - \mubar_k}{l_k(\bSbar_i)} \right] \\
&\varepsilon_{1,ct,2,k} = n^{-1}\E\abs{ K_h(\bSbar_{ii}) \left\{ I(T_i=k)-\pi_k(\bSbar_i)\right\} \ipwbari \frac{Y_i - \mubar_k}{l_k(\bSbar_i)} } \\
&\varepsilon_{2,ct,2,k} = n^{-1}\E\left( \left[K_h(\bSbar_{ji}) \left\{ I(T_j=k)-\pi_k(\bSbar_j)\right\} \ipwbari \frac{Y_i - \mubar_k}{l_k(\bSbar_i)} \right]^2\right)^{1/2}.
\end{align}

We now evaluate each term.  The first term can be simplified through change-of-variables:
\begin{align}
&m_{1,ct,2,k}(\bZ_j) = \E_{\bSbar_i}\left[ K_h(\bSbar_{ji})\left\{ I(T_j = k)-\pi_k(\bSbar_j)\right\} \frac{\E(Y_i -\mubar_k\mid \bSbar_i, T_i = k)}{l_k(\bSbar_i)}\right] \\
&= \int K_h(\bSbar_j - \bs_1)\xi_k(\bs_1)\left\{ I(T_j = k) - \pi_k(\bSbar_j)\right\} d\bs_1 \\
&= \int K(\bpsi_j)\xi_k(h\bpsi_j + \bSbar_j)\left\{ I(T_j = k) - \pi_k(\bSbar_j)\right\} d\bpsi_j \\
&= \int K(\bpsi_j)\left\{ \xi_k(\bSbar_j) + h\bpsi_j\trans\ddbs\xi_k(\bSbar_j) + \ldots + \frac{h^q}{q!}\bpsi_j^{\otimes q} \otimes \ddbsq \xi_k(\bSbar_j^*)\right\}\left\{ I(T_j = k) - \pi_k(\bSbar_j)\right\} d\bpsi_j \\
&= \left\{ \xi_k(\bSbar_j) + \frac{h^q}{q!}\int K(\bpsi_j)\bpsi_j^{\otimes q} \otimes\ddbsq\xi_k(\bSbar_j^*)d\bpsi_j\right\}\left\{ I(T_j = k) - \pi_k(\bSbar_j)\right\} \\
&= \left\{\frac{I(T_j=k)}{\pi_k(\bSbar_j)}-1\right\}\E(Y_i -\mubar_k\mid \bSbar_i=\bSbar_j, T_i = k) \\
&\qquad + \frac{h^q}{q!}\int K(\bpsi_j)\bpsi_j^{\otimes q} \otimes\ddbsq\xi_k(\bSbar_j^*) d\bpsi_j\left\{ I(T_j = k) - \pi_k(\bSbar_j)\right\},
\end{align}
where $\norm{\bSbar_j^*-\bSbar_j} \leq h\norm{\bpsi_j}$ and:
\begin{align}
\xi_k(\bs) = \frac{\E(Y_i -\mubar_k\mid \bSbar_i=\bs, T_i = k)}{\pi_k(\bs)}.
\end{align}

For the second term, due to the centering:
\begin{align}
&m_{2,ct,2,k}(\bZ_i) = \E_{\bSbar_j}\left[ K_h(\bSbar_{ji}) \left\{ \pi_k(\bSbar_j)-\pi_k(\bSbar_j)\right\} \ipwbari \frac{Y_i - \mubar_k}{l_k(\bSbar_i)} \right] = 0 \\
&m_{ct,2,k}(\bZ_i) = \E(m_{2,ct,2,k}(\bZ_i)) = 0.
\end{align}

For the remaining terms:
\begin{align}
&\varepsilon_{1,ct,2,k} = n^{-1}h^{-2}K(\bzero)\E\abs{ \left\{ I(T_i=k)-\pi_k(\bSbar_i)\right\} \ipwbari \frac{Y_i - \mubar_k}{l_k(\bSbar_i)} } = O(n^{-1}h^{-2})\\
&\varepsilon_{2,ct,2,k} = n^{-1}\E\left( \left[K_h(\bSbar_{ji}) \left\{ I(T_j=k)-\pi_k(\bSbar_j)\right\} \ipwbari \frac{Y_i - \mubar_k}{l_k(\bSbar_i)} \right]^2\right)^{1/2} = O(n^{-1}h^{-2}),
\end{align}
where the order of the second error can be obtained from bounding terms inside the expectation.  Now, we apply the projection lemma to find that:
\begin{align}
\Wscrtld_{ct,2,k} &= -n^{1/2} \left[ n^{-1}\sum_{j=1}^n m_{1,ct,2,k}(\bZ_j) -m_{ct,2,k} + n^{-1}\sum_{j=1}^n m_{1,ct,2,k}(\bZ_j) -m_{ct,2,k} \right.\\
&\qquad \left. + m_{ct,2,k}+ O_p(\varepsilon_{1,ct,2,k} + \varepsilon_{2,ct,2,k})\vphantom{[ n^{-1}\sum_{j=1}^n}\right] \\
&= n^{-1/2}\sum_{j=1}^n -\left\{ \frac{I(T_j = k)}{\pi_k(\bSbar_j)}-1 \right\}\E(Y_j - \mubar_k\mid \bSbar_j, T_j = k) + O_p(h^q) + O_p(n^{-1/2}h^{-2}).
\end{align}
We used that $\pi(\bs)$ and $E(Y|\bSbar=\bs,T=k)$ are $q$-times continuously differentiable to bound the remainder error term from $m_{1,ct,2,k}(\bZ_j)$.

We now repeat a similar analysis for $\Wscrtld_{nc,2,k}$.  Let:
\begin{align}
&m_{1,nc,2,k}(\bZ_j) = \E_{\bZ_i}\left[ K_h(\bSbar_{ji}) \left\{ \pi_k(\bSbar_j)-\pi_k(\bSbar_i)\right\} \ipwbari \frac{Y_i - \mubar_k}{l_k(\bSbar_i)} \right]\\
&m_{2,nc,2,k}(\bZ_i) = \E_{\bZ_j}\left[ K_h(\bSbar_{ji}) \left\{ \pi_k(\bSbar_j)-\pi_k(\bSbar_i)\right\} \ipwbari \frac{Y_i - \mubar_k}{l_k(\bSbar_i)} \right] \\
&m_{nc,2,k} = \E\left[ K_h(\bSbar_{ji}) \left\{ \pi_k(\bSbar_j)-\pi_k(\bSbar_i)\right\} \ipwbari \frac{Y_i - \mubar_k}{l_k(\bSbar_i)} \right] \\
&\varepsilon_{1,nc,2,k} = n^{-1}\E\abs{ K_h(\bSbar_{ii}) \left\{ \pi_k(\bSbar_i)-\pi_k(\bSbar_i)\right\} \ipwbari \frac{Y_i - \mubar_k}{l_k(\bSbar_i)} } \\
&\varepsilon_{2,nc,2,k} = n^{-1}\E\left( \left[K_h(\bSbar_{ji}) \left\{ \pi_k(\bSbar_j)-\pi_k(\bSbar_i)\right\} \ipwbari \frac{Y_i - \mubar_k}{l_k(\bSbar_i)} \right]^2\right)^{1/2}.
\end{align}

The first term is:
\begin{align}
&m_{1,nc,2,k}(\bZ_j) = E_{\bSbar_i}\left[ K_h(\bSbar_{ji}) \left\{ \pi_k(\bSbar_j)-\pi_k(\bSbar_i)\right\} \frac{\E(Y_i - \mubar_k\mid \bSbar_i,T_i = k}{l_k(\bSbar_i)} \right] \\
& = \int K_h(\bSbar_j - \bs_1) \left\{ \pi_k(\bSbar_j) - \pi_k(\bs_1)\right\} \xi_k(\bs_1) d\bs_1 \\
&= \int K(\bpsi_j)\left\{ \pi_k(\bSbar_j) - \pi_k(h\bpsi_j + \bSbar_j)\right\} \xi_k(h\bpsi_j + \bSbar_j) d\bpsi_j \\
&= \int K(\bpsi_j)\left\{ -h\bpsi_j\trans \ddbs \pi_k(\bSbar_j) - \ldots - \frac{h^q}{q!}\bpsi_j^{\otimes q}\otimes\ddbsq\pi_k(\bSbar_j^*)\right\} \xi_k(h\bpsi_j + \bSbar_j) d\bpsi_j \\
&= O_p(h^q)
\end{align}
where $\bSbar_j^*$ is such that $\norm{\bSbar_j^* - \bSbar_j} \leq h\norm{\bpsi_j}$ and the last equality can be obtained through bounding $\ddbsq\pi_k(\bSbar_j^*)$ and $\xi_k(h\bpsi_j+\bSbar_j)$. Similarly, for the second term:
\begin{align}
&m_{2,nc,2,k}(\bZ_i) = \E_{\bSbar_j}\left[ K_h(\bSbar_{ji}) \left\{ \pi_k(\bSbar_j)-\pi_k(\bSbar_i)\right\} \ipwbari \frac{Y_i - \mubar_k}{l_k(\bSbar_i)} \right] \\
&=\int K_h(\bs_2 - \bSbar_i)\left\{ \pi_k(\bs_2)-\pi_k(\bSbar_i)\right\} \ipwbari \frac{Y_i - \mubar_k}{l_k(\bSbar_i)} f(\bs_2)d\bs_2 \\
&=\int K(\bpsi_i) \left\{ \pi_k(h\bpsi_i+\bSbar_i)-\pi_k(\bSbar_i)\right\} f(h\bpsi_i+\bSbar_i) d\bpsi_i\ipwbari \frac{Y_i - \mubar_k}{l_k(\bSbar_i)} \\
&=\int K(\bpsi_i) \left\{ h\bpsi_i\trans \ddbs\pi_k(\bSbar_i) + \ldots + \frac{h^q}{q!}\bpsi_i^{\otimes q}\otimes \ddbsq \pi_k(\bSbar_i^*)\right\} f(h\bpsi_i+\bSbar_i) d\bpsi_i\ipwbari \frac{Y_i - \mubar_k}{l_k(\bSbar_i)} \\
&=O_p(h^q),
\end{align}
where $\bSbar_i^*$ is such that $\norm{\bSbar_i^* - \bSbar_i} \leq h \norm{\bpsi_i}$ and the last equality could be obtained through bounding $\ddbsq\pi_k(\bSbar_i^*)$ and $f(h\bpsi_i + \bSbar_i)$.  The errors are:
\begin{align}
&\varepsilon_{1,nc,2,k} = n^{-1}\E\abs{ K_h(\bzero) 0 \ipwbari \frac{Y_i - \mubar_k}{l_k(\bSbar_i)} } =0\\
&\varepsilon_{2,nc,2,k} = n^{-1}\E\left( \left[K_h(\bSbar_{ji}) \left\{ \pi_k(\bSbar_j)-\pi_k(\bSbar_i)\right\} \ipwbari \frac{Y_i - \mubar_k}{l_k(\bSbar_i)} \right]^2\right)^{1/2}= O(n^{-1}h^{-2}),
\end{align}
where the order of the second error can be obtained from bounding terms inside the expectation.  Application of the projection lemma now yields:
\begin{align}
\Wscrtld_{nc,2,k} &= -n^{1/2}\left[n^{-1}\sum_{i=1}^n \left\{ m_{1,nc,2,k}(\bZ_j) - m_{nc,2,k}\right\} + n^{-1}\sum_{i=1}^n \left\{ m_{2,nc,2,k}(\bZ_j) - m_{nc,2,k}\right\}\right] \\
&\qquad + m_{nc,2,k} + O_p(\varepsilon_{1,nc,2,k}+ \varepsilon_{2,nc,2,k}) \\
&= O_p(h^q) -n^{1/2}m_{nc,2,k} + O_p(n^{-1/2}h^{-2}),
\end{align}
where we use that $Var\{m_{1,nc,2,k}(\bZ_j)\}= O(h^{2q})$ and $Var\{m_{2,nc,2,k}(\bZ_i)\}= O(h^{2q})$ and apply Lemma \ref{lm:mnord}.  We now evaluate $m_{nc,2,k}$:
\begin{align}
&m_{nc,2,k} = \E\left[ K_h(\bSbar_{ji}) \left\{ \pi_k(\bSbar_j)-\pi_k(\bSbar_i)\right\} \frac{E(Y_i\mid \bSbar_i,T_i=k) - \mubar_k}{l_k(\bSbar_i)} \right] \\
&= \iint K_h(\bs_2 - \bs_1) \left\{ \pi_k(\bs_2)-\pi_k(\bs_1)\right\} \xi_k(\bs_1) f(\bs_2)d\bs_2 d\bs_1 \\
&= \iint K(\bpsi_1) \left\{ \pi_k(h\bpsi_1 + \bs_1)-\pi_k(\bs_1)\right\} \xi_k(\bs_1) f(h\bpsi_1 + \bs_1)d\bpsi_1 d\bs_1\\
&= \iint K(\bpsi_1) \left\{ h\bpsi_1\trans\ddbs\pi_k(\bs_1) + \ldots + \frac{h^q}{q!}\bpsi_1^{\otimes q} \otimes \ddbsq\pi_k(\bs_1^*)\right\} f(h\bpsi_1 + \bs_1)d\bpsi_1 \xi_k(\bs_1) d\bs_1 \\
&= O_p(h^q),
\end{align}
where $\bs^*$ is such that $\norm{\bs^* - \bs_1} \leq h\norm{\bpsi_1}$, and the last equality follows from bounding $\ddbsq\pi_k(\bs_1^*)$ and $f(h\bpsi_1 + \bs_1)$.  We have now have that:
\begin{align}
\Wscrtld_{nc,2,k} = O_p(n^{1/2}h^q) + O_p(n^{-1/2}h^{-2}).
\end{align}

We now proceed to expand $\Wscrtld_{3,k}$.  
We then first analyze the gradients in general, under model $\Mscr_{\pi}$, and under model $\Mscr_{\pi}\cap \Mscr_{\mu}$, using Lemma \ref{lm:avgrd}.  First note that:
\begin{align}
\Wscrtld_{3,k} &= n^{-1/2}\sum_{i=1}^n \left\{ \ddbalphT \frac{1}{\pihat_k(\bX_i;\balphbar,\bbetabar)}(\balphhat-\balphbar) + \ddbbetaT \frac{1}{\pihat_k(\bX_i;\balphbar,\bbetabar)}(\bbetahat-\bbetabar)\right\}I(T_i = k)(Y_i - \mubar_k) \\
&\qquad + O_p\left\{ n^{1/2}\left(\norm{\balphhat-\balphbar}^2 + \norm{\bbetahat-\bbetabar}^2 + \norm{\balphhat-\balphbar}\norm{\bbetahat - \bbetabar}\right)\right\},
\end{align}
using that $\ddbalphT\pihat_k(\bX_i;\bthet)^{-1}$ and $\ddbbetaT\pihat_k(\bX_i;\bthet)^{-1}$ are Lipshitz continuous in $\bthet$.  Now it can be shown that $\P\{\balphhat_{\Ascrpi^c} = \bzero\} \to 1$ and $\P\{ \bbetahat_{\Ascrmu^c} = \bzero\} \to 1$, using arguments from \cite{hui2015tuning} and \cite{zou2009adaptive} when working models are correctly specified.  It can also be shown that this still holds under misspecified models, provided that the least false parameters $\balphbar$ and $\bbetabar$ exist and are sparse, using arguments similar to those from \cite{lu2012robustness} and \cite{zou2009adaptive}.  Let:
\begin{align}
\bu_{k,n} = n^{-1}\sum_{i=1}^n \ddbalph \frac{I(T_i=k)(Y_i-\mubar_k)}{\pihat_k(\bX_i;\balphbar,\bbetabar)} \text{ and }\bv_{k,n} = n^{-1}\sum_{i=1}^n \ddbbeta \frac{I(T_i=k)(Y_i-\mubar_k)}{\pihat_k(\bX_i;\balphbar,\bbetabar)}.
\end{align}

Using that $\P\{\balphhat_{\Ascrpi^c} = \bzero\} \to 1$ and $\P\{ \bbetahat_{\Ascrmu^c} = \bzero\} \to 1$, we have that $\bu_{k,n,\Ascrpi^c}\trans n^{1/2}(\balphhat-\balphbar)_{\Ascrpi^c} = o_p(1)$, $\bv_{k,n,\Ascrmu^c}\trans n^{1/2}(\bbetahat-\bbetabar)_{\Ascrmu^c}=o_p(1)$, $n^{1/2}\norm{(\balphhat-\balphbar)_{\Ascrpi^c}}^2=o_p(1)$, and $n^{1/2}\norm{(\bbetahat-\bbetabar)_{\Ascrmu^c}}^2=o_p(1)$, so that:
\begin{align*}
\Wscrtld_{3,k} &= \bu_{k,n,\Ascrpi}\trans n^{1/2}(\balphhat-\balphbar)_{\Ascrpi} + \bu_{k,n,\Ascrpi^c}\trans n^{1/2}(\balphhat-\balphbar)_{\Ascrpi^c} \\
&\qquad + \bv_{k,n,\Ascrmu}\trans n^{1/2}(\bbetahat-\bbetabar)_{\Ascrmu} + \bv_{k,n,\Ascrmu^c}\trans n^{1/2}(\bbetahat-\bbetabar)_{\Ascrmu^c} \\
&\qquad + O_p\left\{ n^{1/2}\left(\norm{\balphhat-\balphbar}^2 + \norm{\bbetahat-\bbetabar}^2 + \norm{\balphhat-\balphbar}\norm{\bbetahat - \bbetabar}\right)\right\} \\
&= \bu_{k,n,\Ascrpi}\trans n^{1/2}(\balphhat-\balphbar)_{\Ascrpi} + \bv_{k,n,\Ascrmu}\trans n^{1/2}(\bbetahat-\bbetabar)_{\Ascrmu} \\
&\qquad + O_p\left\{ n^{1/2}\left(\norm{(\balphhat-\balphbar)_{\Ascrpi}}^2 + \norm{(\bbetahat-\bbetabar)_{\Ascrmu}}^2 + \norm{(\balphhat-\balphbar)_{\Ascrpi}}\norm{(\bbetahat - \bbetabar)_{\Ascrmu}}\right)\right\}.
\end{align*}

Applying Lemma \ref{lm:avgrd} to the gradient restricted to the respective active sets:
\begin{align}
&n^{-1}\sum_{i=1}^n \left\{\ddbalphT \frac{I(T_i =k)(Y_i - \mubar_k)}{\pihat_k(\bX_i;\balphbar,\bbetabar)}\right\}_{\Ascrpi} \\
&\qquad = \E\left[ \Kdot_h(\bSbar_{ji})\trans \left\{ \pi_k(\bSbar_i) - I(T_j=k)\right\}\frac{I(T_i =k)(Y_i - \mubar_k)}{\pi_k(\bSbar_i)l_k(\bSbar_i)}\bXdagT_{ji}\right]_{\Ascrpi} + O_p(b_n) \\
&\qquad = \bu_{k,\Ascrpi}\trans + O_p(b_n)\\
&n^{-1}\sum_{i=1}^n \left\{\ddbbetaT \frac{I(T_i =k)(Y_i - \mubar_k)}{\pihat_k(\bX_i;\balphbar,\bbetabar)} \right\}_{\Ascrmu}\\
&\qquad = \E\left[ \Kdot_h(\bSbar_{ji})\trans \left\{ \pi_k(\bSbar_i) - I(T_j=k)\right\}\frac{I(T_i =k)(Y_i - \mubar_k)}{\pi_k(\bSbar_i)l_k(\bSbar_i)}\bXddagT_{ji}\right]_{\Ascrmu} + O_p(b_n)\\
&\qquad = \bv_{k,\Ascrmu}\trans + O_p(b_n).
\end{align}

We now verify that $\bu_{k,\Ascrpi}$ and $\bv_{k,\Ascrmu}$ are $O_p(1)$ in general.  First note that:
\begin{align}
&\bu_{k,\Ascrpi}\trans = \E\left( \Kdot_h(\bSbar_{ji})\trans \frac{\pi_k(\bSbar_i) - I(T_j=k)}{l_k(\bSbar_i)}\left[\E(Y_i - \mubar_k \mid \bSbar_i,T_i=k)\bXdagT_{j} \right.\right.\\
&\qquad\qquad\qquad \left.\left. - \E\left\{(Y_i - \mubar_k) \bXdagT_i \mid \bSbar_i,T_i = k\right\}\right]\vphantom{\E\Kdot_h(\bSbar_{ji})}\right)_{\Ascrpi} \\
&= \E\left( \frac{\Kdot_h(\bSbar_{ji})\trans}{f(\bSbar_i)} \left[\E(Y_i - \mubar_k \mid \bSbar_i,T_i=k)\E(\bXdagT_j \mid \bSbar_j)- \E\left\{(Y_i - \mubar_k) \bXdagT_i \mid \bSbar_i,T_i = k\right\}\right] \right.\\
&\left. - \Kdot_h(\bSbar_{ji})\trans\frac{\pi_k(\bSbar_j)}{l(\bSbar_i)} \left[\E(Y_i - \mubar_k \mid \bSbar_i,T_i=k)\E(\bXdagT_j \mid \bSbar_j, T_j=k) - \E\left\{(Y_i - \mubar_k) \bXdagT_i \mid \bSbar_i,T_i = k\right\}\right]\vphantom{\frac{\Kdot_h(\bSbar_{ji})}{l_k(\bSbar_i)}}\right)_{\Ascrpi} \\
&= \sum_{u=1}^4 \E \left\{ \Kdot_h(\bSbar_{ji})\trans \frac{\bdeta_{1,u,k}(\bSbar_i)}{f(\bSbar_i)}\frac{\bdeta_{2,u,k}(\bSbar_j)}{f(\bSbar_j)}\right\}_{\Ascrpi},
\end{align}
where:
\begin{align}
&\bdeta_{1,1,k}(\bs) = \E(Y_i - \mu_k \mid \bSbar_i = \bs,T_i=k) \qquad \bdeta_{2,1,k}(\bs) = \E(\bXdagT_j\mid \bSbar_j=\bs)f(\bs)\\
&\bdeta_{1,2,k}(\bs) = \E\left\{(Y_i - \mu_k)\bXdagT_i \mid \bSbar_i = \bs,T_i=k\right\} \qquad \bdeta_{2,2,k}(\bs) = f(\bs) \\
&\bdeta_{1,3,k}(\bs) = \frac{\E(Y_i - \mu_k \mid \bSbar_i = \bs,T_i=k)}{\pi_k(\bs)} \qquad \bdeta_{2,3,k}(\bs) = l_k(\bs)\E(\bXdagT_j\mid \bSbar_j=\bs,T_j=k) \\
&\bdeta_{1,4,k}(\bs) = \frac{\E(Y_i - \mu_k \mid \bSbar_i = \bs,T_i=k)}{\pi_k(\bs)} \qquad \bdeta_{2,4,k}(\bs) = l_k(\bs).
\end{align}
Each of the four terms in $\bu_{k,\Ascrpi}\trans$ can be simplified through change-of-variables:
\begin{align}
&\E \left\{ \Kdot_h(\bSbar_{ji})\trans \frac{\bdeta_{1,u,k}(\bSbar_i)}{f(\bSbar_i)}\frac{\bdeta_{2,u,k}(\bSbar_j)}{f(\bSbar_j)}\right\} = \iint \Kdot_{h}(\bs_{21})\trans \bdeta_{1,u,k}(\bs_2)\bdeta_{2,u,k}(\bs_2)d\bs_1d\bs_2 \\
&\qquad\qquad = \iint h^{-1}\Kdot(\bpsi_2)\trans \bdeta_{1,u,k}(h\bpsi_2 + \bs_2)\bdeta_{2,u,k}(\bs_2)d\bpsi_2d\bs_2.
\end{align}
For some vector $\bu$, let $\Kdot(\bu) = (\Kdot(\bu)_{\{1\}}, \Kdot(\bu)_{\{2\}})\trans$ be the partial derivatives of $K(\bu)$ with respect to the first and second components of $\bu$, evaluated at $\bu$.  Similarly, for some $\bs_1$ and $\bs_2$, let $\left\{\bdeta(\bs_1)_{1,u,k}\bdeta(\bs_2)_{1,u,k}\right\}_{\{(i,j)\}}$ denote the $(i,j)$-th element of $\bdeta(\bs_1)_{1,u,k}\bdeta(\bs_2)_{1,u,k}$ evaluated at $\bs_1$ and $\bs_2$, for $i=1,2$ and $j=1,\ldots,p+1$.  Applying integration by parts, the $j$-th element of the above expectation is:
\begin{align}
&\sum_{i=1}^2 \iint h^{-1}\Kdot(\bpsi_2)_{\{i\}} \left\{ \bdeta_{1,u,k}(h\bpsi_2 + \bs_2)\bdeta_{2,u,k}(\bs_2)\right\}_{\{ (i,j)\}}d\bpsi_2d\bs_2 \\
&\qquad = \sum_{i=1}^2 \int h^{-1} K(\bpsi_2)\left\{ \bdeta_{1,u,k}(h\bpsi_2 + \bs_2)\bdeta_{2,u,k}(\bs_2)\right\}_{\{ (i,j)\}}\Big|_{\bpsi_2}d\bs_2 \\
&\qquad\qquad - \iint K(\bpsi_2)\ddpsitwoi\left\{ \bdeta_{1,u,k}(h\bpsi_2 + \bs_2)\bdeta_{2,u,k}(\bs_2)\right\}_{\{ (i,j)\}}d\bpsi_2d\bs_2 \\
&\qquad = - \sum_{i=1}^2 \iint K(\bpsi_2)\ddpsitwoi\left\{ \bdeta_{1,u,k}(h\bpsi_2 + \bs_2)\bdeta_{2,u,k}(\bs_2)\right\}_{\{ (i,j)\}}d\bpsi_2d\bs_2 \\
&\qquad = O(1),
\end{align}
where the second to last and last equalities can be shown by bounding terms using that $\E(Y_i \mid \bSbar_i=\bs,T_i=k)$, $\pi_k(\bs)$, $f(\bs)$, $\E(Y_i\bX_i\mid \bSbar_i=\bs,T_i=k)$, $\E(\bX_i\mid\bSbar_i = \bs,T_i=k)$ are differentiable in $\bs$ for $k=0,1$, $\E(\bX\mid \bSbar=\bs)$ is continuous in $\bs$, $\Xscr$ is compact, and $K(\bu)$ is a kernel function.  Consequently:
\begin{align}
\bu_{k,\Ascrpi}\trans = \sum_{u=1}^4 \E \left\{ \Kdot_h(\bSbar_{ji})\trans \frac{\bdeta_{1,u,k}(\bSbar_i)}{f(\bSbar_i)}\frac{\bdeta_{2,u,k}(\bSbar_j)}{f(\bSbar_j)}\right\}_{\Ascrpi} = O(1).
\end{align}
Applying the same argument it can be shown that $\bv_{k,\Ascrmu}\trans = O(1)$ for $k=0,1$ as well.  

We now consider simplifying $\bv_{k,\Ascr_{\beta}}\trans$ under $\Mscr_{\pi}$.  First we note that under $\Mscr_{\pi}$, $T \indep \bX\mid \balphbar\trans\bX$.  This implies that $T\indep\bX \mid \bSbar$.  Applying this and similar calculations used above for $\bu_{k,\Ascrpi}\trans$:
\begin{align} \label{e:vkunderps}
&\bv_{k,\Ascr_{\beta}}\trans
&=\E\left[ \Kdot_h(\bSbar_{ji})\trans\frac{\pi_k(\bSbar_i)-\pi_k(\bSbar_j)}{l_k(\bSbar_i)}\left\{\E(Y_i\mid \bSbar_i,T_i=k)\E(\bXddagT_j \mid \bSbar_j)-\E\left(Y_i\bXddagT_i\mid \bSbar_i,T_i=k\right)\right\}\right]_{\Ascrmu}.
\end{align}
We now evaluate: 
\begin{align}
&\bv_{k,\Ascrmu}\trans = \Bigg[\iint \Kdot_h(\bs_{21})\trans\frac{\pi_k(\bs_1) - \pi_k(\bs_2)}{\pi_k(\bs_1)}f(\bs_2)\left\{\E(Y_i\mid \bSbar_i = \bs_1,T_i=k)\E(\bXddagT_j\mid \bSbar_j=\bs_2) \right.\\
&\qquad\qquad\qquad\left. -\E\left( Y_i \bXddagT_i \mid \bSbar_i=\bs_1, T_i =k\right)\vphantom{\E}\right\}d\bs_2d\bs_1 \Bigg]_{\Ascrmu}\\
&= \Bigg[\iint h^{-1}\Kdot(\bpsi_1)\trans \frac{\pi_k(\bs_1) - \pi_k(h\bpsi_1+\bs_1)}{\pi_k(\bs_1)}f(h\bpsi_1+\bs_1)\left\{\E(Y_i\mid \bSbar_i = \bs_1,T_i=k) \right.\\
&\qquad\left. \E(\bXddagT_j\mid \bSbar_j=h\bpsi_1+\bs_1)-\E\left( Y_i \bXddagT_i \mid \bSbar_i=\bs_1, T_i =k\right)\vphantom{\E}\right\}d\bpsi_1d\bs_1 \Bigg]_{\Ascrmu}\\
&= \Bigg[-\iint \Kdot(\bpsi_1)\trans \frac{\bpsi_1\trans \ddbs\pi_k(\bs_1)+h\bpsi_1^{\otimes 2}\otimes \ddbssq\pi_k(\bs_1^*)}{\pi_k(\bs_1)}f(h\bpsi_1+\bs_1)\left\{\E(Y_i\mid \bSbar_i = \bs_1,T_i=k) \right.\\
&\qquad\left. \E(\bXddagT_j\mid \bSbar_j=h\bpsi_1+\bs_1)-\E\left(Y_i \bXddagT_i \mid \bSbar_i=\bs_1, T_i =k\right)\vphantom{\E}\right\}d\bpsi_1d\bs_1 \Bigg]_{\Ascrmu}\\
&= \Bigg[-\iint \Kdot(\bpsi_1)\trans \frac{\bpsi_1\trans \ddbs\pi_k(\bs_1)}{\pi_k(\bs_1)}f(\bs_1)\left\{\E(Y_i \mid \bSbar_i = \bs_1,T_i=k) \right.\\
&\qquad\left. \E(\bXddagT_j\mid \bSbar_j=\bs_1)-\E\left( Y_i \bXddagT_i \mid \bSbar_i=\bs_1, T_i =k\right)\vphantom{\E}\right\}d\bpsi_1d\bs_1 \Bigg]_{\Ascrmu}+O(h),
\end{align}
where $\bs^*$ is such that $\norm{\bs^* - \bs_1} \leq h\norm{\bpsi_1}$ and we use that $f(s)$ and $\E(X\mid \bSbar=\bs)$ are continuously differentiable and that $\pi_k(\bs)$ is twice continuously differentiable, $\Xscr$ is compact, $\Kdot(\bu)$ is bounded and integrable, to bound terms in the remainder.  After some re-arrangement, this can be further simplified:
\begin{align}
&\bv_{k,\Ascrmu}\trans = \Bigg\{-\int \frac{\ddbsT\pi_k(\bs_1)}{\pi_k(\bs_1)}\int\bpsi_1\Kdot(\bpsi_1)\trans d\bpsi_1f(\bs_1)\left[\E(Y_i \mid \bSbar_i = \bs_1,T_i=k) \right.\\
&\qquad\qquad\left. \E(\bXddagT_j\mid \bSbar_j=\bs_1)-\E\left( Y_i \bXddagT_i \mid \bSbar_i=\bs_1, T_i =k\right)\vphantom{\E}\right]d\bs_1 \Bigg\}_{\Ascrmu}+O(h)\\
&= \E\left[\frac{\ddbsT\pi_k(\bSbar_i)}{\pi_k(\bSbar_i)}\left\{\E(Y_i \mid \bSbar_i,T_i=k)\E(\bXddagT_i\mid \bSbar_i)-\E\left( Y_i \bXddagT_i \mid \bSbar_i, T_i =k\right)\vphantom{\E}\right\}\right]_{\Ascrmu} +O(h) \label{e:vsimp}\\
&= \bzero + O(h),
\end{align}
where the second equality follows from that $\int\bpsi_1 \Kdot(\bpsi_1)\trans d\bpsi_1 = -\bI_{2\times 2}$ by integration by parts.  Let the partial derivatives of $\pi_k(\bs)$ with respect to $\bs$, evaluated at $\bs$, be denoted by $\partial\pi_k(\bs)/\partial\bs\trans = (\partial\pi_k(\bs)/\partial s_1, \partial\pi_k(\bs)/\partial s_2)$.  Under $\Mscr_{\pi}$ when the PS model is correct, $\partial\pi_k(\bs)/\partial s_2 = 0$ since $\pi_k(\bs)$ would depend only on the first argument.  The last equality follows from noting this and that the first row of $\bXddagT_i$ is $\bzero\trans$.

Finally, we consider the case under $\Mscr_{\pi} \cap \Mscr_{\mu}$.  In this case we have not only that $T\indep \bX \mid \bSbar$ but also $\E(Y\mid\bSbar,T=k,\bX)=g_{\mu}(\betabar_0+\betabar_1 k + \bbetabar\trans\bX)=\E(Y\mid\bSbar,T=k)$.  Thus in this case:
\begin{align}
\E(Y_i \bXdagT \mid \bSbar_i, T_i = k) &= \E(Y_i \mid \bSbar_i, T_i = k)\E(\bXdagT\mid \bSbar_i).
\end{align}
Consequently, continuing from an analogous expression for $\bu_{k,\Ascrpi}$ from \eqref{e:vkunderps}:
\begin{align}
&\bu_{k,\Ascrpi}\trans = \\
&\E\left[ \Kdot_h(\bSbar_{ji})\trans\frac{\pi_k(\bSbar_i)-\pi_k(\bSbar_j)}{l_k(\bSbar_i)}\E(Y_i - \mubar_k\mid \bSbar_i,T_i=k)\left\{\E(\bXdagT_j \mid \bSbar_j)-\E(\bXdagT_i \mid \bSbar_i)\right\}\right]_{\Ascrpi}.
\end{align}
Evaluating the expression, we obtain that:
\begin{align}
\bu_{k,\Ascrpi}\trans &= \Bigg\{\iint \Kdot_h(\bs_{21})\trans \frac{\pi_k(\bs_1)-\pi_k(\bs_2)}{\pi_k(\bs_1)} \E(Y_i - \mubar_k\mid\bSbar_i=\bs_1,T_i=k)\\
&\qquad\qquad \left\{ \E(\bXdagT_j\mid \bSbar_j = \bs_2) - \E(\bXdagT_i\mid \bSbar_i = \bs_1)\right\}f(\bs_2)d\bs_2 d\bs_1 \Bigg\}_{\Ascrpi} \\
&= \Bigg\{\iint h^{-1}\Kdot(\bpsi_1)\trans \frac{\pi_k(\bs_1)-\pi_k(h\bpsi_1+\bs_1)}{\pi_k(\bs_1)} \E(Y_i - \mubar_k\mid\bSbar_i=\bs_1,T_i=k)\\
&\qquad\qquad \left\{ \E(\bXdagT_j\mid \bSbar_j = h\bpsi_1+\bs_1) - \E(\bXdagT_i\mid \bSbar_i = \bs_1)\right\}f(h\bpsi_1+\bs_1)d\bpsi_1 d\bs_1 \Bigg\}_{\Ascrpi} \\
&= \Bigg\{-h\iint \Kdot(\bpsi_1)\trans \frac{\bpsi_1\trans\ddbs\pi_k(\bs_1^*)}{\pi_k(\bs_1)} \E(Y_i - \mubar_k\mid\bSbar_i=\bs_1,T_i=k)\\
&\qquad\qquad \left\{ \bpsi_1\otimes\ddbs\E(\bXdagT_j\mid \bSbar_j = \bs_1^{**}) \right\}f(h\bpsi_1+\bs_1)d\bpsi_1 d\bs_1 \Bigg\}_{\Ascrpi} \\
&= O(h),
\end{align}
where $\bs_1^*$ and $\bs_1^{**}$ are values such that $\norm{\bs_1^* - \bs_1} \leq h\norm{\bpsi_1}$ and $\norm{\bs_1^{**} - \bs_1} \leq h\norm{\bpsi_1}$.  The last equality can be shown by bounding terms inside the integral by using that $\pi_k(\bs)$ is continuously differentiable and bounded away from $0$, $\E(Y - \mubar_k\mid \bSbar =\bs, T = k)$ is continuous, $\E(\bX\mid\bSbar=\bs)$ is continuously differentaible, $f(\bs)$ is continuous, and $\Xscr$ is compact.  The same argument can be applied to show that $\bv_{k,\Ascrmu}\trans=O(h)$ for $k=0,1$, under $\Mscr_{\pi}\cap\Mscr_{\mu}$.

We now collect all the results in the main expansion:
\begin{align}
\Wscrhat_k &= n^{-1/2}\sum_{i=1}^n \frac{I(T_i =k)}{\pi_k(\bSbar_i)}(Y_i - \mubar_k) -\left\{\frac{I(T_i =k)}{\pi_k(\bSbar_i)}-1\right\}\E(Y_i - \mubar_k\mid \bSbar_i,T_i=k)  \\
&\qquad + \bu_{k,\Ascrpi}\trans n^{1/2}(\balphhat-\balphbar)_{\Ascrpi} + \bv_{k,\Ascrmu}\trans n^{1/2}(\bbetahat-\bbetabar)_{\Ascrmu} \\
&\qquad + O_p(b_n) + O_p(n^{1/2}h^q + n^{-1/2}h^{-2})+ O_p(h^q + n^{-1/2}h^{-2})+O_p(n^{1/2}a_n^2) \\
&= n^{-1/2}\sum_{i=1}^n \frac{I(T_i =k)Y_i}{\pi_k(\bX_i;\bthetbar)} -\left\{\frac{I(T_i =k)}{\pi_k(\bX_i;\bthetbar)}-1\right\}\E(Y_i\mid \balphbar\trans\bX_i,\bbetabar\trans\bX_i,T_i=k) -\mubar_k \\
&\qquad + \bu_{k,\Ascrpi}\trans n^{1/2}(\balphhat-\balphbar)_{\Ascrpi} + \bv_{k,\Ascrmu}\trans n^{1/2}(\bbetahat-\bbetabar)_{\Ascrmu} + O_p(n^{1/2}h^q + n^{-1/2}h^{-2})
\end{align}
where $\bu_{k,\Ascrpi}\trans$ and $\bv_{k,\Ascrmu}$ are deterministic vectors such that, for $k=0,1$, $\bv_{k,\Ascrmu}=\bzero$ under $\Mscr_{\pi}$ and $\bu_{k,\Ascrpi}=\bv_{k,\Ascrmu}=\bzero$ under $\Mscr_{\pi}\cap\Mscr_{\mu}$.  The final form of the expansion by using that $\balphhat_{\Ascrpi}$ and $\bbetahat_{\Ascrmu}$ admit an asymptotically linear expansion, using arguments similar to those from \cite{zou2009adaptive} and \cite{lu2012robustness} so that:
\begin{align*}
n^{1/2}(\balphhat-\balphbar)_{\Ascrpi} = n^{-1/2}\sum_{i=1}^n \bPsi_{i,\Ascrpi} + o_p(1) \text{ and } n^{1/2}(\bbetahat - \bbetabar)_{\Ascrmu} = n^{-1/2}\sum_{i=1}^n \bUpsi_{i,\Ascrmu} + o_p(1),
\end{align*}
where $\bPsi_{i,\Ascrpi} = \E(\bU_{\balph,\Ascrpi}\bU_{\balph,\Ascrpi}\trans)^{-1}\bU_{\balph,i,\Ascrpi}$ and $\bUpsi_{i,\Ascrmu} = \E(\bU_{\bbeta,\Ascrmu}\bU_{\bbeta,\Ascrmu}\trans)^{-1}\bU_{\bbeta,i,\Ascrmu}$, $\bU_{\balph,i} = \bX_i\left\{ T_i - \pi_1(\bX_i;\alphbar_0,\balphbar)\right\}$, and $\bU_{\bbeta,i} = \bX_i\left\{ Y_i - \mu_{T_i}(\bX_i;\betabar_0,\betabar_1,\bbetabar)\right\}$, $\pi_1(\bx;\alpha_0,\balph) = g_{\pi}(\alpha_0+\balph\trans\bx)$, and $\mu_k(\bx;
\beta_0,\beta_1,\bbeta)=g_{\mu}(\beta_0+\beta_1 k + \bbeta\trans\bx)$.

We next verify Theorem 2 to characterize the influence function contribution from estimating the PS.  Since $[\bU_{\balph,\Ascrpi}]$ is a finite dimensional subspace of $\Lscr_{2}^0$ spanned by the components of $\bU_{\balph,\Ascrpi}$, the projection of $\varphi_{i,k}$ onto it is given by population least squares:
\begin{align*}
\Pi\left\{ \varphi_{i,k}\mid [\bU_{\balph,\Ascrpi}]\right\} = \E(\varphi_{i,k}\bU_{\balph,\Ascrpi}\trans)\E\left( \bU_{\balph,\Ascrpi}\bU_{\balph,\Ascrpi}\trans\right)^{-1}\bU_{\balph,\Ascrpi}.
\end{align*}
As discussed above, $n^{1/2}(\balphhat-\balphbar)_{\Ascrpi} = \E( \bU_{\balph,\Ascrpi}\bU_{\balph,\Ascrpi}\trans)^{-1}n^{-1/2}\sum_{i=1}^n\bU_{\balph,i,\Ascrpi} + o_p(1)$.
It thus suffices to show that $\bu_{k,\Ascrpi}\trans = -\E(\varphi_{i,k}\bU_{\balph,\Ascrpi}\trans) + o(1)$.

We proceed by simplifying the covariance term:
\begin{align}
&\E(\varphi_{i,k}\bU_{\balph,\Ascrpi}\trans) = \E\left(\left[ \frac{I(T_i=k)Y_i}{\pi_k(\bX_i;\bthetbar)} - \left\{ \frac{I(T_i =k)}{\pi_k(\bX_i;\bthetbar)}-1\right\}\E(Y_i\mid \bSbar_i,T_i=k)-\mubar_k\right]\bU_{\balph,\Ascrpi}\trans\right) \\
&= \E\left(\left[ \frac{I(T_i=k)Y_i}{\pi_k(\bX_i;\bthetbar)} - \left\{ \frac{I(T_i =k)}{\pi_k(\bX_i;\bthetbar)}-1\right\}\E(Y_i\mid \bSbar_i,T_i=k)\right]\bX_{i,\Ascrpi}\trans\left\{ T_i - \pi_1(\bX_i;\alphbar_0,\balphbar)\right\}\right).
\end{align}

First consider the $k=1$ case.  Using that $\pi_1(\bX;\alphbar_0,\balphbar)=\pi_1(\bSbar)$ and $T\indep\bX\mid \bSbar$ under $\Mscrpi$:
\begin{align}
&\E(\varphi_{i,1}\bU_{\balph,\Ascrpi}\trans) = \E\left(\left[ \frac{I(T_i=1)Y_i}{\pi_1(\bX_i;\bthetbar)} - \left\{ \frac{I(T_i =1)}{\pi_1(\bX_i;\bthetbar)}-1\right\}\E(Y_i\mid \bSbar_i,T_i=1)\right]\bX_{i,\Ascrpi}\trans T_i \right) \\
&- \E\left(\left[ \frac{I(T_i=1)Y_i}{\pi_1(\bX_i;\bthetbar)} - \left\{ \frac{I(T_i =1)}{\pi_1(\bX_i;\bthetbar)}-1\right\}\E(Y_i\mid \bSbar_i,T_i=1)\right]\bX_{i,\Ascrpi}\trans \pi_1(\bX_i;\alphbar_0,\balphbar) \right) \\
&= \E\left[ \E\left( Y_i \bX_{i,\Ascrpi}\trans \mid \bSbar_i, T_i=1\right) - \left\{1-\pi_1(\bSbar_i)\right\}\E(\bX_{i,\Ascrpi}\trans\mid \bSbar_i,T_i=1) \E(Y_i \mid \bSbar_i,T_i=1)\right] \\
&- \E\left[ \E\left( Y_i \bX_{i,\Ascrpi}\trans \mid \bSbar_i, T_i=1\right)\pi_1(\bSbar_i) - \left\{\pi_1(\bSbar_i)-\pi_1(\bSbar_i)\right\}\E(\bX_{i,\Ascrpi}\trans\mid \bSbar_i,T_i=1) \E(Y_i \mid \bSbar_i,T_i=1)\right] \\
&=\E\left( \left\{ 1-\pi_1(\bSbar_i)\right\}\left\{ \E(Y_i\bX_{i,\Ascrpi}\trans\mid \bSbar_i,T_i=1)-\E(\bX_{i,\Ascrpi}\trans\mid \bSbar_i)\E(Y_i \mid \bSbar_i,T_i=1)\right\}\right) \\
&=-\bu_{1,\Ascrpi}\trans + O(h),
\end{align}
where $\bu_{k,\Ascrpi}\trans$ has same form as derived in \eqref{e:vsimp} for $\bv_{k,\Ascrmu}$, except that $\bXddagT_i$ is replaced by $\bXdagT_i$.

In the $k=0$ case, again using $\pi_1(\bX;\alphbar_0,\balphbar)=\pi_1(\bSbar)$ and $T\indep\bX\mid \bSbar$ under $\Mscrpi$:
\begin{align}
&\E(\varphi_{i,0}\bU_{\balph,\Ascrpi}\trans) = \E\left(\left[ \frac{I(T_i=0)Y_i}{\pi_0(\bX_i;\bthetbar)} - \left\{ \frac{I(T_i =0)}{\pi_0(\bX_i;\bthetbar)}-1\right\}\E(Y_i\mid \bSbar_i,T_i=0)\right]\bX_{i,\Ascrpi}\trans T_i \right) \\
&- \E\left(\left[ \frac{I(T_i=0)Y_i}{\pi_0(\bX_i;\bthetbar)} - \left\{ \frac{I(T_i =0)}{\pi_0(\bX_i;\bthetbar)}-1\right\}\E(Y_i\mid \bSbar_i,T_i=0)\right]\bX_{i,\Ascrpi}\trans \pi_1(\bX_i;\alphbar_0,\balphbar) \right) \\
&= \E\left\{ \E(Y_i\mid \bSbar_i,T_i=0)\E(\bX_{i,\Ascrpi}\trans \mid \bSbar_i)\pi_1(\bSbar_i)\right\} -\E\left\{\E(Y_i\bX_{i,\Ascrpi}\trans\mid \bSbar_i,T_i=0)\pi_1(\bSbar_i)\right\} \\
&= \E\left\{ \pi_1(\bSbar_i)\E(Y_i\mid \bSbar_i,T_i=0)\E(\bX_{i,\Ascrpi}\trans\mid \bSbar_i) - \E(Y_i\bX_{i,\Ascrpi}\trans\mid \bSbar_i,T_i=0)\right\} \\
&= - \bu_{0,\Ascrpi}\trans + O(h).
\end{align}

\section*{Web Appendix C: Covariate Correlation Matrix in Simulations}
As described in the Simulation Study Section 4.1, the covariates were generated as $\bX = diag(\bSigtld^{-1/2})(\bXtld - \mutld)$, and here we report the values of its covariance matrix $diag(\bSigtld^{-1/2}) \bSigtld diag(\bSigtld^{-1/2})$ for each group of $15$ covariates. The covariates are ordered as ulcerative colitis disease subtype, female gender, use of anti-TNF therapy, use of immunomodulator, primary sclerosing cholangitis (PSC), elevated C-reactive protein, race1, race2, counts of ever smoking from NLP, counts of current smoking from NLP, counts of never smoking from NLP, utilization score, disease duration, and age. For simulations when $p>15$, we used block diagonal matrix where we repeated this correlation structure for each group of $15$ covariates.

\setcounter{MaxMatrixCols}{15}
$$\mbox{\footnotesize$
\begin{bmatrix}{}
  1.00 & -0.01 & -0.16 & -0.14 & 0.08 & -0.03 & 0.01 & 0.01 & 0.02 & -0.05 & 0.05 & 0.08 & -0.01 & -0.05 & 0.08 \\ 
  -0.01 & 1.00 & 0.00 & -0.02 & -0.08 & -0.02 & 0.02 & 0.02 & -0.01 & -0.02 & 0.02 & 0.02 & 0.04 & 0.05 & 0.11 \\ 
  -0.16 & 0.00 & 1.00 & 0.32 & -0.02 & 0.19 & -0.12 & 0.01 & -0.05 & -0.03 & 0.05 & -0.16 & 0.02 & 0.04 & 0.07 \\ 
  -0.14 & -0.02 & 0.32 & 1.00 & 0.02 & 0.21 & -0.17 & 0.00 & -0.07 & -0.02 & 0.08 & -0.17 & 0.02 & 0.04 & 0.10 \\ 
  0.08 & -0.08 & -0.02 & 0.02 & 1.00 & 0.02 & 0.01 & 0.03 & -0.01 & -0.00 & 0.02 & -0.02 & -0.01 & 0.00 & 0.05 \\ 
  -0.03 & -0.02 & 0.19 & 0.21 & 0.02 & 1.00 & -0.21 & 0.03 & -0.05 & -0.09 & 0.13 & -0.10 & 0.06 & 0.08 & 0.14 \\ 
  0.01 & 0.02 & -0.12 & -0.17 & 0.01 & -0.21 & 1.00 & 0.00 & 0.04 & 0.07 & -0.17 & 0.10 & -0.04 & -0.05 & -0.12 \\ 
  0.01 & 0.02 & 0.01 & 0.00 & 0.03 & 0.03 & 0.00 & 1.00 & -0.08 & -0.00 & -0.03 & -0.07 & -0.01 & 0.01 & 0.08 \\ 
  0.02 & -0.01 & -0.05 & -0.07 & -0.01 & -0.05 & 0.04 & -0.08 & 1.00 & 0.01 & -0.07 & -0.03 & -0.09 & -0.08 & -0.08 \\ 
  -0.05 & -0.02 & -0.03 & -0.02 & -0.00 & -0.09 & 0.07 & -0.00 & 0.01 & 1.00 & -0.16 & 0.01 & -0.03 & -0.04 & -0.08 \\ 
  0.05 & 0.02 & 0.05 & 0.08 & 0.02 & 0.13 & -0.17 & -0.03 & -0.07 & -0.16 & 1.00 & 0.25 & 0.19 & 0.17 & 0.28 \\ 
  0.08 & 0.02 & -0.16 & -0.17 & -0.02 & -0.10 & 0.10 & -0.07 & -0.03 & 0.01 & 0.25 & 1.00 & 0.34 & 0.19 & 0.21 \\ 
  -0.01 & 0.04 & 0.02 & 0.02 & -0.01 & 0.06 & -0.04 & -0.01 & -0.09 & -0.03 & 0.19 & 0.34 & 1.00 & 0.88 & 0.18 \\ 
  -0.05 & 0.05 & 0.04 & 0.04 & 0.00 & 0.08 & -0.05 & 0.01 & -0.08 & -0.04 & 0.17 & 0.19 & 0.88 & 1.00 & 0.18 \\ 
  0.08 & 0.11 & 0.07 & 0.10 & 0.05 & 0.14 & -0.12 & 0.08 & -0.08 & -0.08 & 0.28 & 0.21 & 0.18 & 0.18 & 1.00 \\ 
  \end{bmatrix}
$}$$

\section*{Web Appendix D: Proportion of Observations with Negative Values of DiPS}

As discussed in Section 2.1, the following reports the proportion of total number of observations ($n$) in simulations 
that have negative values of DiPS, by varying size of $p$ and model specification scenario.  The results are 
averaged over $R=1,000$ repetitions.
\bigskip

\scalebox{1} {
  \centering
  \bgroup
  \def\arraystretch{.6}%
  \setlength\tabcolsep{.5em}
	  \begin{tabular}{ccccc}
	  \toprule
	  & \textbf{p} & \textbf{Both Correct} & \textbf{Misspecified $\mu_k(\bx)$} & \textbf{Misspecified $\pi_1(\bx)$} \\
	  \midrule
	  \multirow{3}[1]{*}{$n=500$} 
	  & 15    & 0.29\% & 0.52\% & 0.04\% \\
	  & 50    & 0.64\% & 0.85\% & 0.08\% \\
	  & 100   & 1.93\% & 2.10\% & 0.40\% \\ \hline
	  \multirow{3}[1]{*}{$n=5000$} 
	  		& 15	& 0.01\% & 0.03\% & 0.01\% \\
	  		& 50    & 0.01\% & 0.03\% & 0.00\% \\
	  		& 100   & 0.01\% & 0.03\% & 0.00\% \\
	  \bottomrule
    \end{tabular}%
    \egroup
	}

\end{document}